\documentclass{article}

\usepackage[OT1]{fontenc}
\usepackage{amsthm,amsmath,amssymb}
\usepackage{natbib}
\usepackage[colorlinks,citecolor=blue,urlcolor=blue]{hyperref}
\usepackage{mathtools}
\usepackage{bbm}
\usepackage{graphicx}
\usepackage{algorithmic}
\usepackage{algorithm}
\usepackage{pdflscape}
\usepackage{multirow}
\DeclareMathOperator*{\argmax}{arg\,max}

\newcommand{\Bayes}{\hat{\theta}^{\text{(Bayes)}}}
\newcommand{\ncm}{\newcommand}
\ncm{\beq}{\begin{equation*}}
\ncm{\eeq}{\end{equation*}}
\ncm{\cD}{{\cal D}}
\ncm{\cDnew}{{\cal D}^{\text{new}}}
\ncm{\cDobs}{{\cal D}^{\text{obs}}}
\ncm{\cDsim}{{\cal D}^{\text{sim}}}
\ncm{\hI}{\hat{I}}
\ncm{\hP}{\hat{P}}
\ncm{\Ipr}{I^\prime}
\ncm{\Om}{\Omega}
\ncm{\om}{\omega}
\ncm{\hatom}{\hat{\omega}}
\ncm{\cL}{{\cal L}}
\ncm{\Prob}{\mathbb{P}}
\ncm{\E}{\mathbb{E}}
\ncm{\I}{\mathbbm{1}}
\newcommand{\mY}{\mathbf{Y}}
\newcommand{\mX}{\mathbf{X}}
\newcommand{\mB}{\mathbf{B}}
\newcommand{\mE}{\mathbf{E}}
\newcommand{\mZ}{\mathbf{Z}}
\newcommand{\mU}{\mathbf{U}}
\newcommand{\mC}{\mathbf{C}}
\newcommand{\mD}{\mathbf{D}}
\newcommand{\mV}{\mathbf{V}}
\newcommand{\mS}{\mathbf{S}}
\newcommand{\mI}{\mathbf{I}}

\newcommand{\Sigmajuv}{\bSigma^\texttt{juv}}
\newcommand{\Sigmaad}{\bSigma^\texttt{ad}}
\ncm{\cX}{\mathcal{X}}

\newcommand{\classI}{\hat{\mathrm{I}}}
\newcommand{\rmI}{\mathrm{I}}
\newcommand{\cI}{\mathcal{I}}
\newcommand{\sfN}{\mathsf{N}}
\newcommand{\cN}{\mathcal{N}}
\newcommand{\sfS}{\mathsf{S}}
\newcommand{\sfK}{\mathsf{K}}

\newcommand{\rmd}{\mathrm{d}}
\newcommand{\bSigma}{\mathbf{\Sigma}}
\newcommand{\wl}{\textit{wing length }}
\newcommand{\nl}{\textit{notch length }}
\newcommand{\np}{\textit{notch position }}
\floatname{algorithm}{Procedure}

\newtheorem{prop}{Proposition}
\newtheorem{lemma}{Lemma}

\begin{document}

% "Title of the paper"
\title{Identification of taxon through classification with partial reject options}

% indicate corresponding author with \corref{}
% \author{\fnms{Måns} \snm{Karlsson}\corref{}\ead[label=e1]{mansk@math.su.se}\thanksref{t1}}
% \thankstext{t1}{Thanks to somebody} 
% \address{line 1\\ line 2\\ printead{e1}}
% \affiliation{Stockholm University}
% 
% \author{\fnms{M\aa ns} \snm{Karlsson}\thanksref{m1}\corref{}\ead[label=e1]{mansk@math.su.se}}
% % \address{\printead{e1}}
% \and
% \author{\fnms{Ola} \snm{H\"ossjer}\thanksref{m1}\ead[label=e2]{ola@math.su.se}}
% % \address{\printead{e2}}
% \affiliation{Stockholm University\thanksmark{m1}}
% 
% \runauthor{M. Karlsson and O. H\"ossjer}
% 
% \address{M. Karlsson\\
% O. H\"ossjer\\
% Department of Mathematics\\ 
% Stockholm University\\ 
% SE - 106 91 Stockholm\\
% Sweden\\
% \printead{e1}\\
% \phantom{E-mail:\ }\printead*{e2}}

\author{M{\aa}ns Karlsson\thanks{Department of Mathematics, Stockholm University, 106 91 Stockholm, Sweden}
\\
Ola H\"{o}ssjer\thanks{Department of Mathematics, Stockholm University, 106 91 Stockholm, Sweden}}
\maketitle

\begin{abstract}
Identification of taxa can significantly be assisted by statistical classification based on trait measurements in two major ways; either individually or by phylogenetic (clustering) methods. In this paper we present a general Bayesian approach for classifying species individually based on measurements of a mixture of continuous and ordinal traits as well as any type of covariates. It is assumed that the trait vector is derived from a latent variable with a multivariate Gaussian distribution. Decision rules based on supervised learning are presented that estimate model parameters through blockwise Gibbs sampling. These decision regions allow for uncertainty (partial rejection), so that not necessarily one specific category (taxon) is output when new subjects are classified, but rather a set of categories including the most probable taxa. This type of discriminant analysis employs reward functions with a set-valued input argument, so that an optimal Bayes classifier can be defined. We also present a way of safeguarding against outlying new observations, using an analogue of a $p$-value within our Bayesian setting. Our method is illustrated on an original ornithological data set of birds. We also incorporate model selection through cross-validation, examplified on another original data set of birds.
\end{abstract}

% \begin{keyword}[class=MSC]
% \kwd[Primary ]{62P10}
% \kwd{62F15}
% \kwd[; secondary ]{62-07}
% \end{keyword}
% 
% \begin{keyword}
% \kwd{Bayesian decision making}
% \kwd{Bayesian classification}
% \kwd{Set-valued reward function}
% \kwd{Partial observations}
% \kwd{Species identification}
% \kwd{Fuzzy classification}
% \end{keyword}

\section{Introduction}

\subsection{Biological taxonomy}

Biological taxonomy relies on correct classification of organisms to taxa, which in turn requires that reliable and detailed data is available to learn the decision rules. The use of traits to numerically classify taxa has a long history, cf.\ 
\citet{SnSo73, Ba79, PaPr80, Fe83} 
for a summary of early work. These traits are either morphological or genetic, and measured on a scale that is either continuous, ordinal or categorical, or a mixture of these. In order to classify new organisms, some kind of reference material (or a training data set) is needed, to which the trait vector of the new subject is compared and matched. This requires that supervised learning or discriminant analysis has been applied to training data in order to generate a decision rule, so that the new organisms can be assigned to the taxon whose decision region includes this subject's trait vector. Such a classifier is either phylogenetic or individual, depending on whether genetic relationships between species are taken into account or not when the decision regions are derived. For instance, in the context of bacterial classification, \citet{Wa07} introduce an individual (naive Bayesian) classifier, whereas \citet{Ga17} define a phylogenetic classifier such that when multiple taxa of the reference material match the trait vector of a new observation, information from a phylogenetic tree \citep{HuRo01} is used in order to facilitate classification.

Population ecology and ornithology in particular are important applications of biological taxonomy. To the authors' knowledge, there has been no published interpretable and general approach to classifying ornithological data. The widely used work of \citet{svensson1992identification} in bird species identification makes rich use of numerical combinations (often refered to as ``wing formulae'') of trait observations to classify birds to taxa, but offers no general solution. Modern approaches to bird species classification include the Merlin Bird ID software, by the Cornell Lab of Ornithology\footnote{See \href{merlin.allaboutbirds.org}{merlin.allaboutbirds.org}.}. Merlin Bird ID uses convolutional neural networks for image recognition \citep{lecun1989backpropagation} and gives the user a ranked list of species to choose from. Due to the machine learning approach parameter interpretability is complicated though, and not accessible to the user.

\subsection{Unsupervised learning}

A number of unsupervised leaning methods methods have been proposed, as summarized for instance in 
\citet{ripley1996pattern, Ha09, Li11}. For continuous traits this includes linear (or Fisher) discriminant analysis \citep[LDA,][]{Fi36}, quadratic discriminant analysis \citep[QDA,][]{smith1946some, bensmail1996regularized, fraley2002model}, heteroscedastic discriminant analysis \citep{KuAn98}, regularized discriminant analysis \citep{Fr89}, mixture discriminant analysis \citep{HaTi96}, and probabilistic Fisher discriminant analysis \citep{BoBr12}. These are all soft discrimination methods \citep{Wa98}, where a joint statistical distribution of class membership (taxon) and trait is assumed, and the class with maximal posterior probability is selected. Bayes' rule is used to compute the posterior probabilities in terms of prior probabilities of classes and the likelihoods of trait vectors conditional on classes. In particular, a naive Bayes classifier corresponds to a model for which the trait vectors are conditionally independent conditional on class \citep{DoPa97}. Logistic regression and penalized logistic regression \citep{Li00} can also be used for soft binary classification by specifying the posterior probabilities directly in terms of a regression model, with traits as covariates. In particular, this makes it easy to incorporate continuous as well as categorical traits. For any of the abovementioned soft classifiers, if the parameters of the likelihood are random as well, the resulting method is an instance of predictive classification or Bayesian discriminant analysis \citep[Section 2.4 of][]{ripley1996pattern, aitchison1980statistical, geisser1993predictive, hjort1986notes}. This includes, for instance, the Bayesian version of QDA  \citep{Geisser1964, fraley2007bayesian}. On the other hand, many machine learning methods, such as support vector machines 
\citep{CoVa95}, distance-weighted discrimination \citep{Ma07}, and boosting \citep{Fr00}, 
require no statistical model at all, since the decision regions of the classifier are defined in terms of an algorithm that does not rely on any model for how data was generated.             

As mentioned above, it is often the case that taxonomic data is a mixture of various types of observations (continuous, integer valued, ordered categorical). For instance, there might be partial observations (due to e.g. rounding or grouping) or some components of the trait vector might be missing. Models for mixed type of data are important for a wide range of applications \citep{de2013analysis}. Here we will use the umbrella term \emph{obfuscated observations} for partial and missing observations among the components of the trait vector. A general way to handle obfuscated data is to define a distance between pairs of trait vectors, and then employ a nearest neighbour method for classfication \citep{Go71,GoLe86,KaRo90,Do21}. Another option for mixed data is to use a model-based soft classifier, with latent multivariate Gaussian variables, whose components are suiably obfuscated in order to generate the observed trait vector. These types of models have been used in Item Response Theory (IRT), originally proposed by \citet{thurstone1925method} in the context of ecudational testing. The theory has since then been developed by many authors \citep{lord1952relation, rasch1993probabilistic, vermunt2001use, lord2008statistical}. More recent Bayesian models of categorical data include \citet{Albert1993, fokoue2003mixtures, chgh05, mcparland2014clustering, mcparland2017clustering}. Another Bayesian soft classifier for ordinal data, a mixed membership model with a gamma distributed hidden variable, has been proposed by \citet{ViGi15}. 

It is often the case that two or more taxa fit the new observation almost equally well. It is possible then to include a reject option of the classifier that corresponds to ``don't know''. Binary classfication with reject options has been studied by \citet{chow1970optimum} and \citet{herbei2006classification} (see also \citet{freund2004generalization}), whereas \citet{ripley1996pattern} adds a reject option for discrimination between an arbitrary number of classes. In taxonomy it is also important to provide some kind of cutoff when none of the candidate taxa fit the trait vector of the new organism well. This is achieved, for model-based soft classifiers, by declaring the newly observed trait vector as an outlier when its $p$-value is smaller than a cutoff. This $p$-value is based on a measure of oytlyingness of the observed new trait vector, using the mixture distribution of traits from training data of all classes as the analogue of a null distribution \citep{ripley1996pattern}.   

\subsection{The present article}

The classification method presented in this paper is a generalization of the approaches suggested in \citet{svensson1992identification} and \citet{Malmhagen2013}. We develop a method that makes it possible to conclude which trait measurements that characterize various taxa, as an alternative to neural networks and other black box methods. Our method can be viewed as an extension of Bayesian QDA that incorporates traits of mixed types, which are components of a hidden Gaussian vector that have been obfuscated in order to obtain the observed trait vector.  In particular we draw on work by \citet{Albert1993, fokoue2003mixtures, mcparland2014clustering, mcparland2017clustering} and defined a Bayesian classifier whose parameters are estimated from training data by means of blockwise Gibbs sampling \citep{geman1984stochastic, robert2013monte}. The inclusion of covariates is mentioned as an extension in the discussion sections of  \citet{mcparland2014clustering, mcparland2017clustering}. Here we contribute to this extension, assuming that the location as well as the scale parameter of the latent Gaussian distribution depend on covariates, so that heteroscedasticity is achieved. We believe this covariate extension is important, since taxonomic trait measurments of population ecology data is usually informative conditional on other information. In other words, the same value of a trait measurement might indicate different taxa for a subject, depending on the particular subject's additional information. In order to control for this additional information, we include it as set of covariates, and construct our classifier conditional on the covariate information for each subject.

A major novelty of our work is the use of set-valued outputs of the classifier. A set of size 0 corresponds to a scenario where no taxon fits the new organism (an outlier), a set of size 1 corresponds the case where one taxon fits the new organism better than all the others, whereas a set of size $N$, where $N$ is the number of classes, corresponds to a reject option. We also allow for decision regions with $2,\ldots,N-1$ species, corresponding to a partial reject option, where no single species is chosen, but some are excluded. In this way we generalize work by \citet{chow1970optimum, ripley1996pattern, freund2004generalization, herbei2006classification}, making use of a reward function with a set-valued input argument, so that an associated set-valued Bayesian classifier is defined. In order to handle outlying trait vectors, in contrast to \citet{ripley1996pattern}, a new organism's trait vector is assigned one $p$-value for each taxon, rather than one single $p$-value for the whole training data set.  

The primary scenario for which we envision use of our proposed method is as follows: A researcher has a desire to construct a procedure to classify subjects into a, usually small, number of taxa. The procedure should not require any advanced or expensive technology, but rather be viable using only visual observations or simple measurements from everyday tools, such as rulers, callipers and scales. Such a procedure has the potential to aid data collection for population ecologists, if high requirements on classification correctness are met. It is assumed that the researcher has access to training data where each subject has a know taxon, and the traits shared by the taxa are observed for the subjects, although not neccessarily perfectly. To achieve this goal, the researcher would apply our proposed methodology and, if feasible, derive a comprehensive list of which trait observations that signify the various taxa, and which are ambiguous. In short, the goal is to specify a decision rule that is not black box, which will be feasible when the number of traits is small.

The secondary scenario for using our methodology is that in which an agent desires to classify an organism to a taxon, but is unable to do so without some kind of aid. This could be due to inexperience of the considered taxa, or because the number of traits and the number of different outcomes of the traits are too large to extract a decision rule that is nimble. In such a situation, our method can provide tailored prediction of a taxon given the observed traits.

 % For instance, an observation of a species outside of it's expected area of occurence has much more impact than an observation of a frequently occuring species. Hence one may prefer that the classifier signals for the rare species when the degree of belief in it is much higher than for the other species.

% An example is the classification of a bird as male or female. For some species, the juvenile males and adult females look very similar, meaning that the same morphological trait measurements will indicate male if the bird is juvenile, and female if it is adult.

Our paper is organized as follows. In Section \ref{case} we present the application with original ornithological data, which will be employed as we present and analyze our method of classification. The core models, multivariate and multiple  regression with Gaussian errors for perfectly observed data, and the corresponding latent Gaussian models for obfuscated data, are presented in Section \ref{modelformulation} in general but concise terms, and specified for our type of data. Section \ref{classification} presents soft (posterior probability) classification in conjunction with reward functions with set-valued input arguments. In Section \ref{modelselection} we propose model selection (both for traits and covariates) through cross-validation and apply it on another original data set. A comparison with \texttt{mclust} \citep{mclust}, the most popular \textsf{R}-package for clustering and classification using discriminant analysis, is made in Section \ref{comparison}. Section \ref{discussion} contains the discussion. Most of the mathematical details are placed in the appendices, and we will refer to them accordingly.

As a final note, we stress that the presented approach in no way is limited to population ecology data and classification of subjects to taxa. It is applicable more generally for classification with partial reject options, of traits vectors whose components might be a mixture of different types.

% Rather, we will focus on the problem of developing a generally applicable statistically rigourous decision rule for the classification task, essentially assuming there is a referential data set. Typical situations where this type of classification occur are in population monitoring, e.g. through bird ringing or various censuses. To further make the scope of this paper managable, we will focus on the use of morphological traits. These are often easy and cheap to register, making it widely accessible for field workers. In its most basic form, it entails merely observing the subject, and classifying based on observed traits.

\section{Bird species classification} \label{case}
In order to showcase in detail the use of our method, we will present an example of bird species classification using an original data set % collected in part by the first author. 
A minor subset of this primary data set was used in \citet{Walinder1988} and \citet{Malmhagen2013}, but the data set as a whole is unpublished. We will also make use of a second original, unpublished data set in Section \ref{modelselection}, the details of which we present in that section. Both data sets were collected by the Falsterbo Bird Observatory.

The primary data set concerns four morphologically similar warblers in the \textit{Acrocephalus} genus; Eurasian Reed Warbler (\textit{Acrocephalus scirpaceus}), Marsh Warbler (\textit{Acrocephalus palustris}), Blyth's Reed Warbler (\textit{Acrocephalus dumetorum}) and Paddyfield Warbler (\textit{Acrocephalus agricola}). They constitute a typical example of species that experts with long experience can classify through visual observation, whereas  those with less experience usually need measurements of particular traits to aid their classification. Our material contains measurements of these traits for birds that have been identified by experts.

Data was collected from wild birds captured and ringed in Falsterbo, Sweden, and from museum specimens at the museums of natural history in Stockholm, Copenhagen, and Tring. All birds have three traits of main concern, the \textit{wing length} \citep[measured as described in][]{svensson1992identification}, the \textit{notch length} of feather P2 \citep[measured as described in Figure 1 of][]{Malmhagen2013} and the relative position of the P2 notch to the rest of the wing, referred to as the \textit{notch position} \citep[taken as described in][]{svensson1992identification}. These are all measurements of different parts of the wing. The only covariate is \textit{age}, and it is binary, with levels \textit{juvenile} and \textit{adult}. \textit{Age} is included as a covariate due to a suspected change in the trait distribution (both location and scale) between the birds' juvenile and adult plumages.

Ideally, \wl and \nl are measured continuously, but in reality this is impossible. Instead, \wl is rounded to integer millimeters, as this is the typical unit that gives consistent measurements. By the same rationale \nl is rounded to half millimeters. Finally, \np is by definition ordered categorical. It measures where in relation to the other wing feathers the notch is positioned and it is defined as the feather closest to the notch.

Overall we have 54 155 observed birds, and these constitute the referential data for this classification problem. These are distributed as presented in Table \ref{acros}. The uneven distribution of observations across taxa will be commented on throughout the analysis of the data set.

% \begin{table}[ht]
% \centering
% \caption{The distribution of the \textit{Acrocephalus} data over species, covariates and traits. Note that the \textit{Paddyfield Warbler} has few observations overall, but almost all of them are complete, whereas the \textit{Reed Warbler} has many observations but a small percentage are complete.} 
% \begin{tabular}{r|r|r|r|r}
% & \multicolumn{4}{|c}{Species}\\
%   \hline
% Covariate & \vtop{\hbox{\strut \textit{Reed}}\hbox{\strut \textit{Warbler}}} & \vtop{\hbox{\strut \textit{Blyth's Reed}}\hbox{\strut \textit{Warbler}}} & \vtop{\hbox{\strut \textit{Paddyfield}}\hbox{\strut \textit{Warbler}}} & \vtop{\hbox{\strut \textit{Marsh}}\hbox{\strut \textit{Warbler}}}\\ 
%   \hline
% \textit{Juvenile} & 36 034 &  41 &  19 & 2 575 \\
% \textit{Adult} & 14 536 &  68 &  12 & 870 \\
% \hline 
% Trait & & & & \\
% \hline
% \textit{Wing length} & 50 559 & 109 &  31 & 3 404 \\ 
% \textit{Notch length}& 1 302 & 109 &  30 & 2 791 \\
% \textit{Notch position} & 487 & 109 &  31 & 456 \\
% \hline
% Tot. no. birds & 50 570 & 109 &  31 & 3 445
% \end{tabular}
% \label{acros}
% \end{table}

\begin{table}[ht]
\centering
\caption{The distribution of the \textit{Acrocephalus} data over species, covariates and traits. Note that the \textit{Paddyfield Warbler} has few observations overall, but almost all of them are complete, whereas the \textit{Reed Warbler} has many observations but only a small percentage of them are complete.} 
{\small
\begin{tabular}{cr|l|l|l|l}
& & \multicolumn{4}{c}{Species}\\
\cline{3-6}
% Trait & Covariate & \vtop{\hbox{\strut \textit{Reed}}\hbox{\strut \textit{Warbler}}} & \vtop{\hbox{\strut \textit{Blyth's Reed}}\hbox{\strut \textit{Warbler}}} & \vtop{\hbox{\strut \textit{Paddyfield}}\hbox{\strut \textit{Warbler}}} & \vtop{\hbox{\strut \textit{Marsh}}\hbox{\strut \textit{Warbler}}}\\ 
Trait & Covariate & Reed W. & Blyth's Reed W. & Paddyfield W. & Marsh W. \\ 
  \hline
Wing  & \textit{Juvenile} & 36 023 &  41 &  19 & 2 534 \\
length & \textit{Adult} & 14 536 &  68 &  12 & 870 \\
\hline 
Notch  & \textit{Juvenile} & 830 &  41 &  18 & 2 212 \\
length & \textit{Adult} & 472 &  68 &  12 & 579 \\
\hline 
Notch & \textit{Juvenile} & 410 &  41 &  19 & 416 \\
Position & \textit{Adult} & 77 &  68 &  12 & 40
\end{tabular}
}
\label{acros}
\end{table}

% Typical for these situations are that measurements come in various types; continuous, integer valued or categorical. Also common is that some of the trait measurements are missing at random for individual birds. Throughout the paper, we will refer to this as obfuscations, and later show that they can all be treated in a uniform way.

% Our task now reduces to that of how to classify a bird that is known to be one of these four species. We can solve this problem on a macro and a micro level, the former being decision regions in the trait space that determine which trait measurements imply a certain species, and the latter being an automatic species prediction for an individual bird, given its measurements.

% There can be big implications of classifying subjects to the wrong taxa, ultimately risking misrepresenting the spatial and/or temporal distribution of the taxa. To mitigate the chances of misclassification, we will present a novel type of decision rules safeguarding against incorrect classification by giving a set of the most likely species. The set is composed in such a way that the probability of the correct species to be in the set is larger than it not being in the set. With this approach, a user may identify cases where an instance of our trait based classification model is insufficient, and choose to use the diffuse result or invest in consulting an expert that can identify the taxa.

\newpage
\section{Model formulation} \label{modelformulation}
We will first present the model for the case where all trait measurements are continuous and perfectly observed and thereafter for the case where the vector of trait measurements are obfuscated in different ways, including missing values.

\subsection{Ideal case; no obfuscated trait measurements} \label{idealmod}
Suppose we have $N$ different categories (or classes), contained in the set $\sfN = \{1,\ldots,N\}$, 
with known prior probabilities $\pi = (\pi_1,\ldots,\pi_N)$. With full data we measure $q$ 
traits and $p$ covariates of each subject. Let $Y_{ijk}$ be the measurement of trait 
$k$ for subject $j$ in category $i$, where $1\leq i\leq N$, $1\leq j\leq n_i$, $1\leq k\leq q$ 
and $n_i$ is the number of subjects in category $i$. We assume that 
\begin{equation} \label{traits}
Y_{ij} = (Y_{ij1},\ldots,Y_{ijq}) \sim \text{N}\left(m_{ij}, \bSigma_{ij}\right)
\end{equation}
are independent random vectors, each one having a multivariate normal distribution, with
\begin{equation*}
m_{ij} = (m_{ij1},\ldots,m_{ijq}) \qquad \text{and} \qquad \bSigma_{ij} = \left(\Sigma_{ijkl}\right)_{k,l=1}^q
\end{equation*}
being the mean vector and the covariance matrix of subject $j$ of category $i$. The Gaussian assumption on $Y_{ij}$ is especially suitable when the inheritable components of the traits are known to be or can be assumed to be of polygenic nature and the environmental influence on the traits consists of many factors, each with a small effect. This is commonly the case for traits such as height or color of body parts \citep{LyWa98, lande2003stochastic}. Let also
\begin{equation*}
x_{ij} = \left(1, x_{ij1},\ldots,x_{ijp}\right) = \left(x_{ijm}\right)_{m=0}^p
\end{equation*}
be the covariate vector of subject $j$ of category $i$. Trait vectors and covariate 
vectors of category $i$ are rows in the matrices $\mY_i = \left(Y_{i1}^\top,\ldots,Y_{in_i}^\top\right)^\top$ 
and $\mX_i= \left(x_{i1}^\top,\ldots,x_{in_i}^\top\right)^\top$  respectively, where $\top$ refers to matrix transposition.  
We now proceed by formulating a multivariate and multiple regression model 
\begin{equation} \label{matmod}
\mY_i = \mX_i\mB_i + \mE_i
\end{equation}
for category $i$, where $\mB_i=\left(B_{imk}; m=0,\ldots,p; k=1,\ldots,q\right)$ is 
the regression parameter matrix, whose first row consists of intercepts for 
the $q$ traits, $m_{ij}$ is the $j^{\text{th}}$ row of $\mX_i\mB_i$, and 
$\mE_i = \left(E_{i1}^\top,\ldots,E_{in_i}^\top\right)^\top$ is an error 
term matrix with independent rows $E_{ij} \sim\text{N}(0,\mathbf{\bSigma}_{ij})$. 

For use in the construction of a joint prior, and later the derivation 
of the marginal posterior distributions of the parameters, the vectorized 
form of our regression model is needed. Denote by $\mbox{vec}(\cdot)$ the vectorization operation of a matrix, that appends the columns of the matrix from left to right on top of each other \citep{macedo2013typing}. The inverse of the vec operation will also be used, and it is denoted by $\mbox{vec}^{-1}(\cdot)$. Then rewrite \eqref{matmod} as 
\begin{equation} \label{U}
\mU_i = \text{vec}(\mY_i) = \mZ_i \beta_i + \text{vec}(\mE_i)
\end{equation}
with $\beta_i = \text{vec}(\mB_i)$. Denoting an identity matrix of rank $q$ 
with $\mI_q$ and using the matrix tensor (or Kronecker) product $\otimes$,
\begin{equation} \label{Z}
\mZ_i =  \mI_q \otimes \mX_i = \begin{pmatrix} \mX_i & 0 & \cdots & 0 \\
                        0 & \mX_i & \ddots & \vdots \\
                        \vdots & \ddots & \ddots & 0 \\
                        0 & \cdots & 0 & \mX_i \end{pmatrix}
\end{equation}
is a block-diagonal matrix with $q$ blocks along the diagonal.

Now suppose we have $A$ covariance classes $\alpha = 1,\ldots,A$ 
for category $i$ such that
\begin{equation} \label{covmats}
\bSigma_{ij} = \bSigma^\alpha_i \quad \text{if } x_{ij} \in \cX^\alpha,
\end{equation}
where $\cX = \cX^1 \cup \ldots \cup \cX^A$ is a disjoint 
decomposition of the predictor space $\cX$. Assume a prior 
$\text{N}\left(\left(b_{i0k},\ldots,b_{ipk}\right)^\top = b_{ik}, \bSigma_{\mB_i}\right)$ on each of the columns $k=1,\ldots,q$ of $\mB_i$. This implies a prior 
$\text{N}\left(\left(b_{i1}^\top,\ldots,b_{iq}^\top\right)^\top = \beta_{i0}, \mI_q \otimes \bSigma_{\mB_i} = \bSigma_{\beta_i}\right)$ 
on $\beta_i$. 
Further, assuming prior independence and imposing an Inverse-Wishart 
distribution $\bSigma_i^\alpha \sim IW(\nu_0, \mathbf{V}_0)$ (see for instance \citet{gelman2013bayesian}) on the covariance matrices in \eqref{covmats} for $\alpha = 1,\ldots,A$, we 
get the joint prior
\begin{equation} \label{jointprior}
p(\beta_i, \bSigma_i^1, \ldots, \bSigma_i^A) = p(\beta_i)\prod_{\alpha=1}^A p(\bSigma_i^\alpha)
\end{equation}
for the parameters of category $i$.

Write $\cD_i=(\mX_i,\mY_i)$ for the training data of category $i$ and let $(x,Y)$ be a new obervation that we want to classify. The predictive classifier in Section \ref{classification} will involve the density $\omega_i = f(Y;x|\cD_i) = \E[f(Y;x,\theta_i|\cD_i)]$ of this new observation, in case it belongs to category $i$, where $f(Y;x,\theta_i)$ is the density function of the trait vector $Y$ conditional on the covariate vector $x$ and the parameter vector $\theta_i$. For a detailed derivation of the collection of model parameters $\theta_i = \left(\mB_i, \bSigma^1_i,\ldots,\bSigma_i^A\right)$ and the posterior category weights $(\omega_1,\ldots,\omega_N)$ we refer to Appendix \ref{ideal}. The Monte Carlo approximations of $\theta_i$ and $\omega_i$, $i \in \sfN$ are
\begin{equation} \label{Bayesest}
\Bayes_i = \frac{1}{R_i}\sum_{r=1}^{R_i} \theta_{ir}
\end{equation}
and
\begin{equation} \label{MComegaY}
\hatom_i = \frac{1}{R_i}\sum_{r=1}^{R_i} f(Y;x,\theta_{ir})
\end{equation}
respectively, where $R_i$ is the number of samples drawn from the posterior distribution of $\theta_i$, with $\theta_{ir}$ the parameter vector obtained by blockwise Gibbs sampling \citep{geman1984stochastic, robert2013monte} in simulation run $r$.

\subsection{General case; all types of obfuscation may occur} \label{realmod}
Overall our setup is the same as in Section \ref{idealmod}, but now we suppose there is only partial information about the complete training data 
set $\cD = \{(\mX_i,\mY_i);\, i=1,\ldots,N\}$. Due to some obfuscation, which could be due to rounding, grouping, categorization or lost measurements of some traits, 
we typically have a mixture of different types of observations \citep{de2013analysis} and only know that 
\begin{equation} \label{obfuscationset}
Y_{ij} \in \sfS_{ij} = \sfS_{ij1}\times \cdots \times \sfS_{ijq},
\end{equation}
i.e. the complete trait vector $Y_{ij}$ for subject $j$ of category 
$i$ is contained in a  hyperrectangle $\sfS_{ij}$, whose components 
are given by $\{\sfS_{ijk}\}_{k=1}^q$. These components are sets, ranging 
in possible size from singletons to infinite intervals of 
$\mathbb{R}$, and they are given by
\begin{equation} \label{obfuscationside}
\sfS_{ijk} = \begin{cases} Y_{ijk}, & k \notin \sfK_{ij}, \\
\left(c_{ijk},d_{ijk}\right], & k \in \sfK_{ij},
\end{cases}
\end{equation}
where $\sfK_{ij}= \left\{k; 1\le k \le q; \, Y_{ijk} \text{ obfuscated}\right\}$.
The obfuscations in \eqref{obfuscationset}-\eqref{obfuscationside} are of the same type as in IRT \citep{lord2008statistical, fox2010bayesian}. The main difference is that the covariance matrix of the trait vector $Y_{ij}$ in \eqref{traits} is arbitrary, whereas the covariance matrix of IRT models typically have a factor analysis structure \citep{spearman1904general}, so called mixture of factor analysis models for mixed data \citep[MFA-MD;][]{mcparland2017clustering}. 

More specifically, the obfuscations in \eqref{obfuscationset}-\eqref{obfuscationside} are of three main types. First, if a trait 
$Y_{ijk}$ is unobserved, written as $Y_{ijk} =  \texttt{NA}$, 
the $k$:th component of $\sfS_{ij}$ is of infinite length; e.g. 
$c_{ijk} = -\infty$, $d_{ijk} = \infty$, and we let the interval 
be open. That is, the interval $\sfS_{ijk}$ equals $\mathbb{R}$.
Secondly, a trait may be obfuscated in such a way that interval 
limits are observed. Rounding is a typical example of this; consider 
a measurement of a trait $y_{ijk} \in \mathbb{R}^+$ that has 
been rounded to $z_{ijk} \in 2\tau\cdot\mathbb{Z}^+$. We put $c_{ijk} = z_{ijk} - \tau$ and $d_{ijk} = z_{ijk} + \tau$, which constitute the limits of the
interval around $z_{ijk}$. Generally, we assume rounding to the midpoint of an 
interval. We can always scale so that the interval is of unit length, 
which would be equivalent to $\tau = 1/2$. Lastly we have the case 
when we observe an ordered categorical random variable $Z_{ijk}$. 
We assume there is an underlying normally distributed variable 
$Y_{ijk}$, and that each category  $z_{ijk}$ corresponds to an 
interval of possible values of $y_{ijk}$. Count data can be treated 
as an instance of this type of obfuscation, for instance traits that represent the   
number of occurrences of something.

We will treat all types of obfuscations in the following unified way.
Suppose trait $k$ of subject $j$ of category $i$ is imperfectly 
observed, i.e. $k \in \sfK_{ij}$. Let $g_{k}$ be the number of 
categories of this trait, which we number as $0,1,\ldots,g_{k}-1$.
The observed category is $z_{ijk} \in \left\{0,1,\ldots,g_{k} - 1\right\}$, 
where $g_{k} =2$ for binary data and $g_{k} = \infty$ for count data. 
The corresponding side of $\sfS_{ij}$ is 
\begin{equation*}
\sfS_{ijk} = \begin{cases}
            \left(-\infty, \frac{1}{2}\right], & \text{if } z_{ijk} = 0, \\
            \left(z_{ijk} - \frac{1}{2}, z_{ijk} + \frac{1}{2}\right], & \text{if } 1\le z_{ijk} \le g_{k}-2, \\
            \left(g_{k}-\frac{3}{2}, \infty\right), & \text{if } z_{ijk} = g_{k} - 1.
          \end{cases}
\end{equation*}
Here, a useful trick would be to add auxiliary categories, that never 
were observed, to take the place of $z_{ijk}=0$ and $z_{ijk}= g_{k} -1$. 
That ensures all observed intervals are of unit length, although we may 
let intervals vary in length if there is reason to construct such a model.
We also write 
\begin{equation*}
Z_{ijk} = z(\sfS_{ijk}) = \begin{cases}
            0, & \text{if } \sfS_{ijk} = \left(-\infty, \frac{1}{2}\right], \\
            \frac{c_{ijk} + d_{ijk}}{2}, & \text{if $\sfS_{ijk}$ is bounded}, \\
            g_{k} - 1, & \text{if } \sfS_{ijk} = \left(g_{k}-\frac{3}{2}, \infty\right],
          \end{cases}
\end{equation*}
for the center point of a finite or half-open and infinite $\sfS_{ijk}$, 
whereas $Z_{ijk} =z\left(\sfS_{ijk}\right) = Y_{ijk}$ when $Y_{ijk} = \sfS_{ijk}$ 
is perfectly observed. We will write the observed training 
data set as
\begin{equation*}
\cDobs = \left\{\left(x_{ij},\sfS_{ij}\right);\, i=1,\ldots,N,j=1,\ldots,n_i\right\}.
\end{equation*}
Finally, we remark on the importance (or lack thereof) of taking 
rounding into account. Consider rounding a Gaussian trait $Y_{ijk}$ 
to $z_{ijk}\in \mathbb{Z}$ for some $k$, $i=1,\ldots,N$ and 
$j=1,\ldots,n_i$. Suppose we have no covariates and that 
$Y_{ijk} \sim N(m_{ik},\sigma_{ik}^2)$ for $j=1,\ldots,n_i$. An 
unbiased estimator of $m_{ik}$ is the average $\bar{Y}_{ik}$, whereas 
$\bar{Z}_{ik}$ is a biased estimator of $m_{ik}$. It is possible to quantify the 
size of the bias using $\sigma_k$ and the width of the rouding 
interval $\mathsf{w} = \eta_k\sigma_k$ \citep{tricker1984effects}.
In short, the larger $\mathsf{w}$ is relative to $\sigma_k$, 
the larger the bias is, as measured by $\eta_k$. Already 
when $\sigma_k=\mathsf{w} = 1$, the bias is very small, 
and hence, unless an extremely precise mean estimate is needed, the bias 
is small compared to the uncertainty of the parameter estimate. Therefore, one might 
regard rounded values as true values, if the standard deviation of the trait that is rounded, is large enough.

Let $(x,\sfS)$ refer to a new observation, for which the trait vector $Y \in \sfS$ is obfuscated for traits $k \in \sfK$. We refer to Appendices \ref{real} and \ref{MCapp} for full details on the derivation of exact estimators of the model parameters $\theta_i$ and posterior weights $\omega_i = \Prob(\sfS;x\mid \mathcal{D}_i^{\text{obs}})$ of category $i$, as well as Monte Carlo approximtions thereof. In brief, for each category $i$ we want to generate the set  
\begin{equation} \label{paramset}
\left\{\theta_{ir}, Y_{ijkr}, \, 1\le j\le n_i, \, k\in\sfK_{ij}; Y_{kr}, k\in\sfK\right\}_{r=1}^{R_i}
\end{equation}
of $R_i$ blockwise Gibbs sampling iterates from the joint density
\begin{equation} \label{MCdens}
p(\theta_i|\cDobs_i)\prod_{j=1}^{n_i} f\left(y_{ij\sfK_{ij}} \mid x_{ij}, \sfS_{ij},\theta_i\right) f\left(y_{\sfK} \mid x, Y_{\sfK^\complement};\theta_i\right)
\end{equation}
of the parameters $\theta_i$, the imperfectly observed training data and the imperfectly observed new data point, where $\theta_{ir} = \left(\beta_{ir}, \bSigma_{ir}^1,\ldots,\bSigma_{ir}^A\right)$, 
$y_{ij\sfK_{ij}r} = \left(y_{ijkr}; k\in \sfK_{ij}\right)$, 
and $y_{\sfK r}=\left(y_{kr} ; k\in\sfK\right)$ refer to the values of these quantities for Monte Carlo iteration $r$, whereas $Y_{\sfK^\complement} = \left(Y_k ; k\notin \sfK\right)$. In particular, $\theta_{ir}$ are drawn from the posterior distribution $p(\theta_i|\cDobs_i)$ of $\theta_i$, conditional on observed data $\cDobs_i$ for category $i$. Whereas the Monte Carlo approximation of the Bayes estimator of $\theta_i$ is given by \eqref{Bayesest}, also for obfuscated data, the Monte Carlo estimator of $\omega_i$ is slightly more complicated than \eqref{MComegaY}, as described in Appendix \ref{realMCapprox}.

\subsection{Example model fit} \label{acromod}
We now have the tools to fit a classification model to our \textit{Acrocephalus} data described in Section \ref{case}. As we can tell from the end of Section \ref{case}, all traits are obfuscated in some way (many are missing), there is one covariate influencing the interpretation of the trait values and we have reason to believe there are different covariance classes.

Our trait vectors $Y_{ij}$ are attributed to species $i=1,\ldots,4$, where $i=1$ corresponds to Eurasian Reed Warbler, $i=2$ to Marsh Warbler, $i=3$ to Paddyfield Warbler and $i=4$ to Blyth's Reed Warbler, whereas $j = 1,\ldots,n_i$ denote individual birds within species. Each $Y_{ij}$ is of length $q=3$, where $Y_{ij1}$ is the \emph{wing length}, $Y_{ij2}$ is the \emph{notch length} and $Y_{ij3}$ is the \emph{notch position}. All traits are obfuscated, but in different ways: $Y_{ij1}$ is continuous, but rounded to nearest integer on a millimeter scale;  $Y_{ij2}$ is continuous, but rounded to nearest half millimeter; whereas $Y_{ij3}$ is ordered categorical.

We have one covariate \emph{age} with values \emph{juveline} or \emph{adult}, which we code as 0 and 1 respectively, so that \textit{juvenile} corresponds to the intercept. We denote this covariate by $x_1$, and it determines to which of the $A=2$ covariance classes each observation belongs. We denote the respective covariance matrices with $\Sigmajuv$ and $\Sigmaad$.

The degrees of freedom hyperparameter of the covariance matrix $\bSigma_i^\alpha$ of the prior is chosen as $\nu_0 = 10$, whereas the scale matrix $V_0$ has a diagonal of 15 and all other elements equal to 5. All covariance matrices have the same prior. Since there is $p=1$ covariate, the matrix used in the construction of the prior on the vectorized regression parameters, $\bSigma_{\mB_i}$, is a diagonal matrix with diagonal $(3,1)$ for all $i$. The mean parameter values $\mB_{i0}=\E(\mB_i)$ of the prior on $\mB_i$ are informative for each $i$ and based on the results in \citet{Malmhagen2013}, as shown in Table \ref{priorparams}.

\begin{table}[ht]
\centering
\caption{Hyperparameter values for the prior of the regression parameter matrices. These values are informed by \citet{Malmhagen2013}, except for the first element of the second row of $\mB_{30}$ and $\mB_{40}$, which we put to $0.75$, since we strongly believe that the pattern of slightly longer wings in adult plumages also hold for these two species.}
{\footnotesize
\begin{align*}
&\text{Reed Warbler} & &\text{Marsh Warbler} \\
\mB_{10} &= \begin{pmatrix}
          67.1 & 11.4 & 108 \\
          0.5 &  1.3 &  3
         \end{pmatrix}  &
\mB_{20} &= \begin{pmatrix}
          70.3 & 9.5 & 105 \\
          0.2 & 0.6 & 1
         \end{pmatrix}  \\
\\
&\text{Paddyfield Warbler} & &\text{Blyth's Reed Warbler} \\
\mB_{30} &= \begin{pmatrix}
          57.4 & 12.7 & 115.3 \\
          0.75 & 1.1 & 0.5
         \end{pmatrix} &
\mB_{40} &= \begin{pmatrix}
          62 & 12.5 & 113 \\
          0.75 & 1.1 & 1
         \end{pmatrix}.
\end{align*}}
\label{priorparams}
\end{table}

% Starting at random points in $\Theta_i$, where $\theta_i$ lives, that are over dispersed compared to the posterior distribution, we would ideally have ended up in the high density region of $\Theta_i$ in the second step, since we sample from explicit posterior distributions. However, since we have large amounts of observations on some of the traits, we would expect the high density region to be small along these axis. Therefore, it might take a long before we find this region, and thus reach stationarity of the Markov Chain. 

% (Should this section be omitted?) As mentioned in Section \ref{realmod}, the bias introduced in the estimate of the mean of a stochastic variable from rounding its outcomes turns negligable when the rounding interval is small enough relative to the variance. We could opt to not take the rounding into regard for the traits \wl and \textit{notch length}, since we have a very small bias for these. Since \np is ordered categorical, we have to treat it as an imperfectly observed component of the underlying multivariate Gaussian trait distribution. We also choose to treat the other traits as imperfectly observed.

Fitting the model in \textsf{R} \citep{r2021}, with all traits regarded as obfuscated, we get the Bayes estimates presented in Table \ref{acroestimates}. Highest posterior density intervals for each parameter are presented in Appendix \ref{hpdint}. Overall, the effect of \textit{age}, our covariate, is to increase the trait values. However, the increase is different across traits and across species.

\begin{table}[!htbp] 
\caption{Bayes estimates $\hat{\mB}_{imk}=\hat{\E}(\mB_{imk}|\cD_i)$ and $\hat{\bSigma}_{ikl}^\alpha=\hat{\E}(\bSigma_{ikl}^\alpha|\cD_i)$ of all parameters of the \emph{Acrocephalus} model (cf. \eqref{Bayesest}), rounded to two decimals, except the \emph{notch position} trait where we present the categories that the regression parameters estimates fall into. The coding of these categories is explained in Table 2 of \citet{Malmhagen2013}.}
\centering
% \begin{tabular}
{\footnotesize
\begin{align*}
&\text{Reed warbler} & &\text{Marsh warbler} \\
\hat{\mB}_{1} &= \begin{pmatrix}
          66.70 & 11.06 & \text{P8/9} \\ 
          0.73 & 1.39 & \text{P9/10}
         \end{pmatrix}  &
\hat{\mB}_{2} &= \begin{pmatrix}
          70.05 & 9.45 & \text{P7} \\ 
          0.69 & 0.61 & \text{P7/8}
         \end{pmatrix}  \\ \\
&\text{Paddyfield Warbler} & &\text{Blyth's Reed Warbler} \\
\hat{\mB}_{3} &= \begin{pmatrix}
          57.31 & 12.67 & \text{T3/T2} \\ 
          0.27 & 1.14 & \text{T2}
         \end{pmatrix} &
\hat{\mB}_{4} &= \begin{pmatrix}
          62.28 & 12.49 & \text{T3} \\ 
          0.01 & 1.13 & \text{T3/T2}
         \end{pmatrix}
\end{align*}
\begin{align*}
&&\text{Reed Warbler}& \\
\hat{\bSigma}_{1}^{\texttt{juv}} &= \begin{pmatrix}
 2.34 & 0.63 & 0.04 \\ 
  0.63 & 0.64 & 0.52 \\ 
  0.04 & 0.52 & 1.46 
  \end{pmatrix} &
\hat{\bSigma}_{1}^{\texttt{ad}} &= \begin{pmatrix}
 2.67 & 0.53 & 0.40 \\ 
  0.53 & 0.64 & 0.45 \\ 
  0.40 & 0.45 & 1.56 
  \end{pmatrix} \\ \\
&&\text{Marsh Warbler}& \\
\hat{\bSigma}_{2}^{\texttt{juv}} &= \begin{pmatrix}
 2.19 & 0.40 & -0.05 \\ 
  0.40 & 0.44 & 0.19 \\ 
  -0.05 & 0.19 & 0.80
  \end{pmatrix} &
\hat{\bSigma}_{2}^{\texttt{ad}} &= \begin{pmatrix}
 2.51 & 0.49 & 0.22 \\ 
  0.49 & 0.53 & 0.18 \\ 
  0.22 & 0.18 & 0.79 
  \end{pmatrix} \\ \\
&&\text{Paddyfield Warbler}& \\
\hat{\bSigma}_{3}^{\texttt{juv}} &= \begin{pmatrix}
 2.27 & 0.27 & -0.43 \\ 
  0.27 & 0.88 & 0.28 \\ 
  -0.43 & 0.28 & 1.37
  \end{pmatrix} &
\hat{\bSigma}_{3}^{\texttt{ad}} &= \begin{pmatrix}
 4.59 & 0.92 & 0.83 \\ 
  0.92 & 1.46 & 0.57 \\ 
  0.83 & 0.57 & 1.31 
  \end{pmatrix} \\ \\
&&\text{Blyth's Reed Warbler}& \\
\hat{\bSigma}_{4}^{\texttt{juv}} &= \begin{pmatrix}
1.98 & 0.41 & 0.16 \\ 
  0.41 & 0.71 & 0.17 \\ 
  0.16 & 0.17 & 0.93
  \end{pmatrix} &
\hat{\bSigma}_{4}^{\texttt{ad}} &= \begin{pmatrix}
1.88 & 0.61 & 0.10 \\ 
  0.61 & 0.85 & 0.32 \\ 
  0.10 & 0.32 & 1.23 
  \end{pmatrix}.
\end{align*}}
\label{acroestimates}
\end{table}

\section{Classification} \label{classification}
We now have interpretable results on the variation of traits, but ultimately we also want to use our knowledge to classify new birds to species. In this section we present classification for the model of Section \ref{realmod}. First, we define more generally the posterior category weights $\omega_i$ that were introduced in Section \ref{idealmod} and then look at canonical classification. Then we introduce set-valued classifiers, including the possibility of classifying empty sets in order to handle outliers. We also showcase the flexibility of the underlying method for classifying among a subset of categories and end with remarks on the choice of the classifier's two tuning parameters $\rho$ and $\tau$.

Let 
\begin{equation}
\cDnew = (x,\sfS)
\label{Dnew}
\end{equation}
denote a new observation with obfuscated traits $\sfK$. We define the posterior weight of category $i$ as
\begin{equation} \label{omega}
\omega_i = \iint_{\sfS} \! f(Y;x,\theta_i) \prod_{k\in \sfK} \mathrm{d}y_{k} \, p(\theta_i \mid  \cDobs_i) \, \mathrm{d}\theta_i,
\end{equation}
where $f$ is the density function of the trait vector $Y=(y_1,\ldots,y_q)$ of the new observation, i.e. the multivariate Gaussian density function. As shown in Appendix \ref{real}, the Markov chain in \eqref{paramset} can be used to find estimates $\hat{\omega}_i$ of these weights. We may then approximate the posterior probability $\hat{p}_i = \hat{\Prob}(I=i \mid  \cDnew, \cDobs)$ of $\cDnew$ to be of category $i$ as
\begin{equation} \label{predprob}
\hat{p}_i = \hat{\Prob}(I=i \mid  \cDnew, \cDobs) = \frac{\pi_i\hat{\omega}_i}{\pi_1\hat{\omega}_1 + \ldots + \pi_N\hat{\omega}_N},
\end{equation}
where $I\in \sfN$ is the true but unknow category of the future observation, with prior distribution $\Prob(I=i)=\pi_i$.

Let $\cN = \mathcal{P}(\sfN) \setminus \emptyset$ denote the 
collection of all non-empty subsets of $\sfN$. Let $\classI \in \cN$ be a classifier 
with $\lvert\classI\rvert \ge 1$. In order to define $\classI$ 
we introduce a reward function $\cN \times \sfN \ni (\cI, i) \mapsto R(\cI,i)$ 
for all $\cI \in \cN$ and $i\in\sfN$. We interpret $R(\cI,i)$ as the reward of a classified set $\cI$ of categories when the true category is $i$. Then put
\begin{align*}
\classI &= \argmax_{\cI} \E \left[R(\cI,\mathrm{I}) \mid \cDobs, \cDnew\right] \\
&= \argmax_{\cI} \sum_{i=1}^N R(\cI,i)p_i
\end{align*}
as the optimal predictive classifier or optimal Bayesian classifier, with the complete training data set $\cDobs$
% \begin{equation} \label{pi}
%  p_i = \Prob(I=i \mid  \cDnew, \cDobs) = \frac{\pi_i\omega_i}{\pi_1\omega_1 + \ldots + \pi_N\omega_N}.
% \end{equation}
and $p_i$ defined as in \eqref{predprob}. Thus, $\hat{\rmI}$ is the set in $\cN$ that maximizes the expected posterior reward. Each classifier $\classI = \classI(\cDobs, \cDnew)$, viewed as a function of a perfectly observed new data point $\cDnew=(x,Y)$, partitions the test data space into decision regions
\begin{equation*}
\Omega_{\cI} = \{(x,Y); \classI = \cI \}
\end{equation*}
for all $\cI \in \cN$. A similar definition of decision regions applies for an obfuscated new data point (\ref{Dnew}), where the trait part $\sfS$ is defined as rectangles of a grid, with $\Omega_{{\cal I}}$ the collection of $(x,\sfS)$ for which $\hat{I}={\cal I}$. For any of these two definitions of decision regions, this gives rise to an \textit{indecisive region}
\begin{align*}
\Lambda = \bigcup_{|\cI| > 1} \Omega_{\cI},
\end{align*}
where we cannot distinguish one particular category with acceptable confidence,
only eliminate some of the categories with low degree of belief.

Whenever the indecisive region $\Lambda$ is nonempty, the classifier allows for a partial reject option. An important special case when this occurs is  Bayesian sequential analysis \citep{degroot1970optimal}, for instance in the context of clinical trials \citep{carlin1998approaches}, before a final decision between the $N$ hypotheses (or categories) has been made. This corresponds to a partition of the test data space $\Omega$ into $N+1$ regions, so that either one of the $N$ hypotheses is chosen ($\hat{\rmI} = \{i\}, i = 1,\ldots,N$) or the decision is postponed to a later time point ($\hat{\rmI} = \{1,\ldots,N\}$).

There is considerable freedom in choosing the reward function $R$. For instance, \citet{chow1970optimum} introduced a reward function for $N=2$ categories. This was generalized in Section 2.1 of \citet{ripley1996pattern} to arbitrary $N \ge 2$, using $R(\{i\},\rmI)=\I(i\in \rmI)$, $R(\{1,\ldots,N\},\rmI)=c$ for some constant $0\le c \le 1$ and $R({\cal I},\rmI) = 0$ for all ${\cal I}$ with $2\le |{\cal I}|\le N-1$. In the next two subsections we will present two reward functions, the first of which corresponds to Ripley's reward function when $|{\cal I}|=1$ and $|{\cal I}|=N$ (using $c=1/N$), but with a nonzero reward for $2\le |{\cal I}|\le N-1$. 
% In the next two subsections we will consider two $R$ functions, the first of which extends Chow's reward function with $c=0.5$ to arbitrary $N$. 

\subsection{Classification to one category}\label{oneclass}
%Present the naive case, i.e. $\rho = 1$ and present the decision regions for our example data, as well as some metrics, like probability to be wrong etc.

Let
\begin{align} 
R(\cI, i) = \begin{cases} 0; & i \notin \cI \\
                             1/|\cI|; & i \in \cI \end{cases} \label{reward1}
\end{align}
which has expected posterior reward
\begin{align*}
\E \left[R(\cI, \rmI) \mid \cDobs, \cDnew \right] = \frac{1}{|\cI|}\sum_{i\in\cI} p_i
\end{align*}
and optimal classifier
\begin{align} \label{class1}
\classI = \{(N)\} = \argmax_i \pi_i \omega_i
\end{align}
where $p_{(1)} < \ldots < p_{(N)}$ are the ordered posterior 
category probabilities. Notice that the indecisive region is empty ($\Lambda = \emptyset$), i.e. 
this reward function leads to a consistent, but potentially 
overzealous, classifier. 

To estimate the probability of classifying wrongly using this approach, we simulated a large number of realisations $Y^\star$ from the predictive posterior distribution of the model in Section \ref{acromod} for each species, under the assumption of a uniform prior distribution over species. Each $Y^\star$ was then attributed to a hyperrectangle in trait space in order to represent obfuscation as described in Section \ref{case}, resulting in half side lengths or $\tau$-values $(1/2, 1/4, 1/2)$ for trait $j=1,2,3$ respectively. We then computed which species an observation each hyperrectangle would predict, using the classifier in \eqref{class1}. This gives a numerical approximation under current obfuscations of the probability of observing a new bird and classifying it wrongly, when using the fitted model on the \emph{Acrocephalus} data:
\begin{equation} \label{prederr1}
\begin{matrix}
\hat{\Prob}\left(\classI \ne \rmI \mid x_1=0 \right) = 0.0251, &
\hat{\Prob}\left(\classI \ne \rmI \mid x_1=1 \right) = 0.0264.
\end{matrix}
\end{equation}
Roughly 1 in 50 birds would be classified erraneously, when birds are distributed uniformly over the species under consideration. In Section \ref{rhoclass} we will show how to reduce this error by allowing for classification to sets of species.

\subsection{Classification with partial reject options} \label{rhoclass}
Choosing the reward function
\begin{align} 
R(\cI, i) = \I_{\{i\in\cI\}} - \rho \lvert \left\{ \iota \in \cI ; \iota\ne(N) \right\} \rvert p_{(N)}, \label{reward2}
\end{align}
we get the expected posterior reward 
\begin{align*} 
\E \left[ R(\cI, \rmI) \mid \cDobs, \cDnew \right] &= \sum_{i \in \cI} p_i - \rho\left(\lvert \cI \rvert - \I_{\{(N) \in \cI \}}\right)p_{(N)},
\end{align*}
which is maximized by
\begin{align} \label{class2}
\classI &= \left\{ i; p_i \ge \rho p_{(N)} \right\} = \left\{ i; \pi_i\omega_i \ge \rho \pi_{(N)}\omega_{(N)} \right\}.
\end{align}
Thus we can tune the risk of classifying wrongly by picking $\rho \in [0,1]$ adequately, 
as it specifies an upper bound on the fraction of the largest 
posterior probability $p_{(N)}$ other posterior probabilities may attain and 
still be excluded. If we choose $\rho =0$, we get the classifier 
$\classI = \sfN$ which means $\Prob(\rmI \in \classI) = 1$ for all 
new observations, but that prediction method does not provide any 
information at all. The other extreme, choosing $\rho = 1$, leads 
to $\classI = \{(N)\}$, and thus our classifier will be the same as
\eqref{class1}. In Section \ref{choosingrho} we present a way of choosing $\rho$ using cross-validation and a maximal accepted misclassification rate. In conclusion, our first classifier is a special case of the second.

Choosing $\rho = 0.1$ yields the exclusion critera $\hat{p}_i < \hat{p}_{(N)}/10$. With this value of $\rho$, we find that the estimated probability of classifying wrongly using the \textit{Acrocephalus} model rounded to four decimals are
\begin{equation*}
\begin{matrix}
\hat{\Prob}\left(\rmI \notin \classI \mid x_1=0 \right) = 0.0058, &
\hat{\Prob}\left(\rmI \notin \classI \mid x_1=1 \right) = 0.0058
\end{matrix}
\end{equation*}
and that the probability of not singling out a particular species is
\begin{equation*}
\begin{matrix}
\hat{\Prob}\left(\lvert\classI\rvert > 1 \mid x_1=0\right) = 0.0790, &
\hat{\Prob}\left(\lvert\classI\rvert > 1 \mid x_1=1\right) = 0.0819.
\end{matrix}
\end{equation*}
This means we have reduced the probability of choosing the wrong species by $76.8\%$ for juvenile birds and $77.9\%$ for adult birds, at a price of not singling out a species with a probability of about 8\% for any covariate value. Of the cases where $\lvert \hat{\rmI} \rvert > 1$, only $0.0044 \%$ will result in a classifier containing three species ($\lvert \hat{\rmI} \rvert = 3$), meaning that we will be able to exclude at least half of the potential species for the vast majority of observations.

\subsection{Classification with reject option}
Following Section 2.1 of \citet{ripley1996pattern}, if none of the $N$ categories support test data $\cDnew$ we would like to include $\emptyset$ as a possible output of the classifier 
$\classI$, so that $\classI \ \in \mathcal{P}(\sfN)$. To this end,
we denote the posterior weight of \eqref{omega} as $\omega_i(x,\sfS)$ 
in order to emphasize its dependence on the test data set (\ref{Dnew}). Then let
\begin{equation} \label{omegabar}
\bar{\omega}_i(x,\sfS) = \iint \! p(\theta_i \mid \cDobs_i)p(\sfS^\prime,x;\theta_i)\, \rmd \theta_i \rmd \sfS^\prime
\end{equation}
where the outer integral is taken over all $\sfS^\prime$ such 
that $\omega_i(x,\sfS^\prime) \le \omega_i(x,\sfS)$. We interpret 
$\bar{\omega}_i(x,\sfS)$ as a $p$-value of test data $(x,\sfS)$ 
for category $i$, i.e. the probability of observing an obfuscated 
trait vector $\sfS^\prime$ of category $i$ with covariate vector 
$x$, whose posterior weight $\omega_i(x,\sfS^\prime)$ is at most 
as large as that of $(x,\sfS)$. As such, it is a measure of the 
degree of outlyingness of $\cDnew$. Our treatment differs from that of \citet{ripley1996pattern} in that we define $N$ distinct $p$-values of test data, one for each category, whereas \citet{ripley1996pattern} defines one single $p$-value for the mixture distribution of all categories. More specifically, given a value of 
$\rho$, we generalize the classifier \eqref{class2} to
\begin{equation} \label{class3}
\classI = \left\{i; \pi_i\omega_i \ge \rho\pi_{(N)}\omega_{(N)} \land \pi_i\bar{\omega}_i \ge \tau\right\}
\end{equation}
where $\bar{\omega}_i = \bar{\omega}_i(x,\sfS)$. Note that 
\eqref{class2} is a special case of \eqref{class3} with $\tau = 0$.

Choosing $\tau = 0.001$ results in
\begin{equation*}
\begin{matrix}
\hat{\Prob}\left(\hat{\rmI} = \emptyset \mid x_1=0 \right) = 6.62\cdot10^{-4}, &
\hat{\Prob}\left(\hat{\rmI} = \emptyset \mid x_1=1 \right) = 6.32\cdot10^{-4}.
\end{matrix}
\end{equation*}
The probability of not choosing a set containing the correct species is
\begin{equation*}
\begin{matrix}
\hat{\Prob}\left(\rmI \notin \classI \mid x_1=0 \right) = 0.0066, &
\hat{\Prob}\left(\rmI \notin \classI \mid x_1=1 \right) = 0.0065,
\end{matrix}
\end{equation*}
whereas the probability of not singling out a particular species ($|\hat{\rmI}|>1$) or getting an outlier ($|\hat{\rmI}|=0$) is
\begin{equation*}
\begin{matrix}
\hat{\Prob}\left( |\hat{\rmI}|\ne 1 \mid x_1 = 0 \right) = 0.0791, &
\hat{\Prob}\left( |\hat{\rmI}|\ne 1 \mid x_1 = 1 \right) = 0.0820.
\end{matrix}
\end{equation*}
These probabilities are very close to the ones in Section \ref{rhoclass}, meaning we can hedge the risk of classifying something that might be a new species (not belonging to $\{1,\ldots,N\}$) entirely at a low cost. The decision regions for these values on $\rho$ and $\tau$ are presented graphically in Appendix \ref{graphs}, where we cover all classification scenarios with missing trait values as well.

It is also possible to include classification with empty outputs in the context of indifference zones \citep{bechhofer1954single, goldsman1986tutorial}. Assume that the parameter space is divided into $N+1$ regions, the first $N$ of which correspond to each of the $N$ hypotheses (or categories), whereas the last region of the parameter space (the indifference zone) corresponds to scenarios where no hypothesis is adequate. Based on this, it is possible to divide the test data space into $N+1$ regions as well, depending on whether the posterior distribution of the parameter puts most of its probability mass in any of the first $N$ parameter regions ($\hat{\rmI} = \{i\}, i =1,\ldots,N$) or in the indifference zone ($\hat{\rmI} = \emptyset$).  

\subsection{Subproblem accessibility}
Having fitted the model to the whole set of species, one may use the fit for any subproblem, e.g. classifying a bird between two species when the others are, for some reason, ruled out. Taking species 1 and 2, i.e. \emph{Eurasian Reed Warbler} and \emph{Marsh Warbler}, we estimate the probability of classifying wrongly and the probability of ending up in the indecisive region $\Lambda$ analogously with Section \ref{rhoclass}. Using $\rho=0.1$ and $\tau = 0$, we find that
\begin{equation*}
\begin{matrix}
\hat{\Prob}\left(\rmI \notin \classI \mid x_1=0 \right) = 0.0037, &
\hat{\Prob}\left(\rmI \notin \classI \mid x_1=1 \right) = 0.000971
\end{matrix}
\end{equation*}
\begin{equation*}
\begin{matrix}
\hat{\Prob}\left(\lvert\classI\rvert > 1 \mid x_1=0\right) = 0.0495, &
\hat{\Prob}\left(\lvert\classI\rvert > 1 \mid x_1=1\right) = 0.0099
\end{matrix}
\end{equation*}
for this particular subproblem.

\subsection{Choosing $\tau$}
A potential risk for misclassification is observing a subject of a 
category not even considered for classification. In order to mitigate this, we introduced $\tau$ as a cut-off value for 
the trait distributions. Indeed, the value of $\tau$ determines
how large deviations in trait measurements we accept without
suspecting that we actually observe a subject from an 
unconsidered category. Choosing $\tau = 0$ allows us to classify 
any point in the whole trait space, i.e. we believe the 
model is perfect in the sense that no unconsidered categories will be observed.

\section{Model selection using cross-validation} \label{modelselection}
Model selection can be used to select covariates and/or traits from a larger set by trading goodness-of-fit against parsimony for each candidate model $m$. Bayesian model selection requires a prior distribution on $m$ \citep{smith1980bayes, kass1995bayes, green1995reversible, tadesse2005bayesian}. This approach involves the computationally intractable likelihood ${\cal L}({\cal D}_i^{obs})$ for each category $i$ and candidate model $m$. Non-Bayesian approaches, on the other hand, typically involve the observed likelihood ${\cal L}(\hat{\theta}_i;{\cal D}_i^{obs})$ for each category $i$ and model $m$, with $\hat{\theta}_i$ the ML-estimator of $\theta_i$, as well as another term that penalizes large models. This includes the AIC \citep{akaike1974a, akaike1998information}, NIC \citep{murata1991criterion}, BIC \citep{schwarz1978estimating, fraley1998many}, and integrated complete-data likelihood \citep[ICL;][]{biernacki2000assessing} criteria, as well as stepwise model selection procedures based on pairwise hypothesis testing \citep{mclachlan2014number}. Here we will use an intermediate approach that retains the Bayesian assumption of random $\theta_i$ for each model $m$ but then compare these models in terms of their predictive performance. To this end we will use $\kappa$-fold cross validation. This will be illustrated on another original, unpublished data set, on two subspecies of Common chiffchaff (\textit{Phylloscopus collybita}), the \textit{collybita} and \textit{abietinus} subspecies. It contains measurements for birds classified to subspecies visually by an expert at Falsterbo Bird Observatory.

\subsection{Cross-validation for our type of models}
The idea of $\kappa$-fold cross-validation is well established, and used for a very wide range 
of model families, see e.g. \citet[p.~256]{wood2017generalized}. 
Choosing $\kappa = n_i$ for category $i$ corresponds to the basic form of 
cross-validation \citep{cover1969learning, stone1974cross}. This procedure is however computationally expensive, because one
has to fit as many models ($=\sum_{i=1}^N n_i$) as there are observations.
Since the method under study is already computationally 
intensive, in particular under widespread obfuscation, 
large $q$ and large data sets, we recommend using $\kappa$-fold 
cross-validation with $\kappa$ a bit smaller, i.e. a non-exhaustive cross-validation.
In an interesting paper \citet{kohavi1995study}
examines cross-validation in general when choosing between classifiers. 
Kohavi concludes that $\kappa < n_i$ is generally preferred when 
picking the best classifier using cross-validation.

To perform $\kappa$-fold cross-validation in general for our class of models, 
we begin by choosing $\kappa \in \mathbb{Z}^+$ independently of $i$. 
Then create fold $l$ for category $i$ by choosing uniformly at random a 
set $J_{il} \subset\{1,\ldots,n_i\}$ comprising $n_i/\kappa$ or $n_i/\kappa +1$ 
observations of $\cDobs_i$, the training data at hand for category $i$. 
Repeat this for all categories until we have left-out test data sets $J_l = \cup_{i=1}^N J_{il}$ for $l = 1,\ldots, \kappa$. Then for each $l$ proceed to fit models on the observations $\cDobs_{(-l)} = \cDobs \setminus \{(x_{ij}, \sfS_{ij}); j \in J_{il}\}_{i=1}^N$ that were not left out and then estimate the posterior category probabilities $\hat{p}_1(\cDobs_{(-l)},(x_{ij}, \sfS_{ij})),\ldots,\hat{p}_N(\cDobs_{(-l)},(x_{ij}, \sfS_{ij}))$ for each observation $(x_{ij}, \sfS_{ij})\in J_l$ that was left out. Choosing a reward function $R$, the corresponding classifier $\hat{\rm{I}}$ may be applied to each set of posterior probabilities in order to generate predictions $\classI_{(-l)ij}$ of the category, or set of categories, of the left-out observations contained in $J_l$. As this is repeated for each $l$, all observations will have been left out at one point when the cross validation is finished.

% Fit the model using training data set $\cD_{i(-l)} = \cD \setminus \ \left\{\left(x_{ij}, Y_{ij}\right); j \in J_{il} \right\}$ 
% and predict the category of all observations in $J_{il}$; denote the 
% classification of individual $j \in J_{il}$ of species $i$ by $\classI_{(-il)ij}$. 
% This procedure is repeated for $N\kappa$ test data sets $\cD_{(-il)}$, 
% when $i=1,\ldots,N$ and $l=1,\ldots,\kappa$. 

To assess the predictive 
performance of the $M$ models under consideration, let $w_i >0$ be weights 
such that $\sum_{i=1}^N w_i = 1$, and compute
\begin{equation} \label{cv}
R^{\texttt{cv}}_m = \sum_{i=1}^N \frac{w_i}{n_i} \sum_{l=1}^\kappa \sum_{j\in J_{il}} R(\classI_{(-l)ij},i),
\end{equation}
with a classifier \eqref{class2} that corresponds to a prespecified value 
of $\rho$ and $\tau$, for $m=1,\ldots,M$, and an appropriately chosen reward function, such as \eqref{reward1} or \eqref{reward2}. One could e.g. 
use the weights $w_i = n_i / \sum_{a=1}^N n_a$ or $w_i = 1/N$, depending 
on whether it is more valuable to be able to predict a category with many observations or 
not. Based on \eqref{cv}, the best classifier is 
\begin{equation*}
\hat{m} = \argmax_{m} \left(R^{\texttt{cv}}_1, \ldots, R^{\texttt{cv}}_M\right).
\end{equation*}
When having computed $R^{\texttt{cv}}_m$, for $m=1,\ldots,M$, one has the 
possibility to choose a more parsimonious model that performes almost as 
well as $\hat{m}$, under the usual tradeoff between simplicity and predictive performance.

\subsection{Choosing traits and covariates} \label{chooseyandx}
We will now examplify usage of cross validation with $\kappa=10$ folds for model selection. At hand we have a data set over two subspecies of Common chiffchaff (\emph{Phylloscopus collybita}). Each subspecies has $q=9$ continuous traits measured. Although they are measured to half millimeters, we will regard them as perfect observations ($\cDobs = \cD$), since the rounding bias is extremely small. Moreover we have two binary covariates (\textit{age} and \textit{season}), and we are interested in which combination of covariates and traits that best predicts the subspecies of a new bird. We will consider models with one covariate, both covariates and both covariates including an interaction term, meaning that the number of predictors $p$ ranges from 1 to 3. 

For each covariate scenario, we fit a model $m$ with one trait, for each trait in turn, and then choose the one with the highest $R^{\texttt{cv}}$ value, based on the reward function \eqref{reward1}. Keeping this trait, we then add another trait and repeat the procedure, and keep the two traits that yield the highest $R^{\texttt{cv}}$ value. The procedure is repeated until we have a model with 8 traits, and for each step, the value of $R^{\texttt{cv}}$ is stored. We used the weights $w_i = n_i / \sum_{k=1}^N n_k$ for all $i$, and put $\rho =1$ and $\tau = 0$, i.e. used the classifier in \eqref{class1}  in the computation of $R^{\texttt{cv}}$.

The main reason for doing this forward selection-like procedure is to reduce the number of models that  are fitted from about $(2^9-1)\cdot 4 \cdot 10 = 20440$ (if every trait combination is tried) to $(9+8+\ldots+1)\cdot 4\cdot 10 = 1800$. Also, for each included trait, if the value of $R^{cv}$ does not increase, we may choose a more parsimonious model. 

Figure \ref{cvplot} shows a plot of how $R^{\texttt{cv}}$ changes with the number of included traits for the various covariate combinations.

\begin{figure}[h!]
  \centering
  \includegraphics[width=\linewidth]{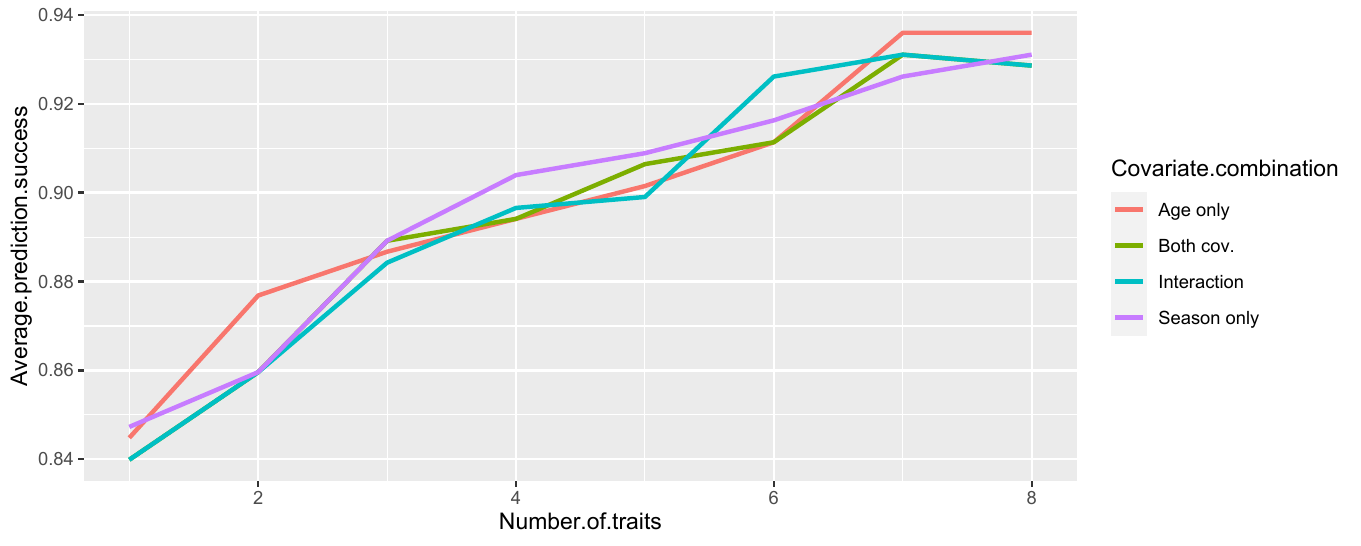} 
  \caption{Each line represents, for a particular covariate combination, the change in average prediction success of the left-out data in the cross validation procedure, as the number of traits increases. In more detail, we use \eqref{reward1} as our reward function and $w_i = 1/N$ for all categories $i$ in \eqref{cv}. No obvious differences between the covariate combinations are visible. It is, however, clear that 3 of the covariate scenarios reach their maximum cross-validation prediction accuracy with 7 traits. The set of included traits are not the same for all covariate scenarios.}
  \label{cvplot}
\end{figure}

\subsection{Choosing $\rho$} \label{choosingrho}
We will now present how the choice of $\rho$, for a classifier $\hat{I}$ based on a reward function \eqref{reward2}, can be made more intuitive by letting the chosen value $\rho_\delta$ correspond to a highest acceptable misclassification rate $\delta$. To this end, we first introduce a reward  function  
\begin{equation} \label{tilder}
 \tilde{R}(\cI,i) = \begin{cases}
             1, \text{ if $i \in \cI$,} \\
             0, \text{ otherwise,}
            \end{cases}
\end{equation}
that outputs 1 when the true category belongs to the classified subset of categories, and otherwise 0. Interpreting $\tilde{R}(\cI,i)=1$ as a correct classification, we estimate the rate of correct classifications through cross-validation as
\begin{equation} \label{mcr}
 R_\rho^{\text{cv}} = \sum_{i=1}^N \frac{w_i}{n_i} \sum_{l=1}^\kappa \sum_{j\in J_{il}} \tilde{R}(\classI_{(-l)ij},i),
\end{equation}
in analogy with \eqref{cv}. In order to speed up computation, we first perform cross-validation to generate category weights $\omega_1(x_{ij},\sfS_{ij}),\ldots,\omega_N(x_{ij},\sfS_{ij})$ for each observaton $(x_{ij},\sfS_{ij})$ in $\cDobs$. We then compute the probabilities $p_1(x_{ij},\sfS_{ij}), \ldots,$ $p_N(x_{ij},\sfS_{ij})$ of all categories according to \eqref{predprob}, given a prior distribution $(\pi_1,\ldots,$ $\pi_N)$ on the categories. The category probabilities can then be used to create classifiers $\hat{I}_{(-l)ij}$ for each observation $(x_{ij},\sfS_{ij})$ with $j\in J_l$, for a grid of $\rho$-values. These classifiers are then inserted into \eqref{tilder}, alongside the correct category, and the misclassification rate $1 - R_\rho^{cv}$ is computed. We may then pick
\begin{equation}
 \rho_\delta = \sup \{\rho^\prime : 1- R_{\rho^\prime}^\text{cv} \le \delta\}.
\end{equation}

To illustrate this, we performed $\kappa = 31$-fold cross-validation for the \textit{acrocephalus} data. Recall that the majority of observations in these data are partial, and thus classification is hard (see e.g. the decision regions in Appendix \ref{graphs}). To compute the category weights, for each fold $J_l$ we used the Bayes estimates $\hat{\theta}_{(-l)i}$ as in \eqref{Bayesest}, based on the observations $(x_{ij},\sfS_{ij})$ of category $i$ not left out ($j\notin J_l$), as parameter values for the trait distributions. Then we numerically integrated
\begin{equation}
 \omega_{\iota}(x_{ij},\sfS_{ij}) = \int_{\sfS_{ij}} f\left(y; x_{ij}, \Bayes_{(-l)\iota}\right) \, \rmd y, \quad \iota=1,\ldots,N,
\end{equation}
for each $(x_{ij},\sfS_{ij})$ with $j\in J_l$ and $l=1,\ldots,\kappa$. We used two different priors in the computation of $p_1(x_{ij},\sfS_{ij}),\ldots,p_N(x_{ij},\sfS_{ij})$; a uniform prior ($\pi_i = 1/N$ for all $i$) and a prior proportional to the number of observations in each category ($\pi_i = n_i/\sum_{\iota=1}^Nn_{\iota}$). As in Section \ref{chooseyandx}, we chose $w_i = n_i / \sum_{k=1}^N n_k$ for all $i$. The classifiers were created according to \eqref{class2} and for $\rho = 0.01, 0.02,\ldots,0.99$ we computed \eqref{mcr} using two different reward functions. First the reward function $\tilde{R}(\cI,i)$, defined in \eqref{tilder}, and then a second reward function
\begin{equation} \label{dotr}
 \dot{R}(\cI,i) = \begin{cases}
             1, \text{ if $i = \cI$,} \\
             0, \text{ otherwise,}
            \end{cases}
\end{equation}
which is more restrictive, as we consider the classifier to be correct when it contains only the correct category. The second reward function \eqref{dotr} is not used for choosing $\rho$, but rather in order to illustrate how its misclassification rate differs from the one obtained from the other less restrictive reward function $\tilde{R}$. The misclassification rates as functions of $\rho$ are presented in Figure \ref{misclassificationrates} for both priors. In Table \ref{spmcr} we give the misclassification rates for each category using the uniform prior, for $\rho=1$ and $\rho = 0.3$.

\begin{figure}[h!] 
  \centering
  \begin{tabular}{cc}
    \includegraphics[width=0.5\linewidth]{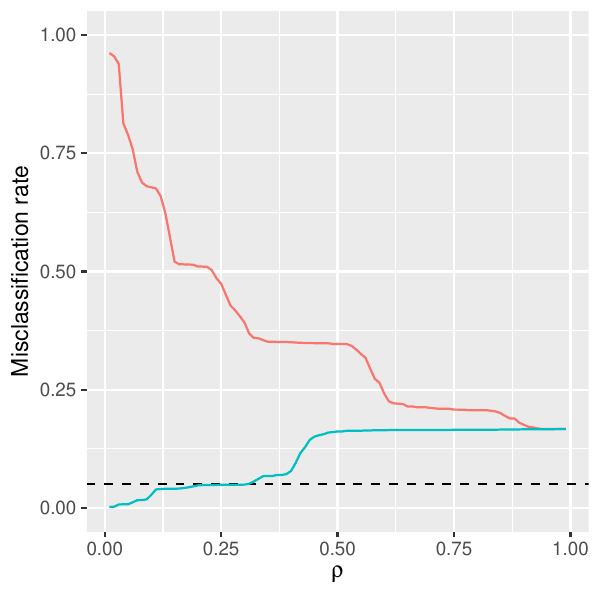} & 
    \includegraphics[width=0.5\linewidth]{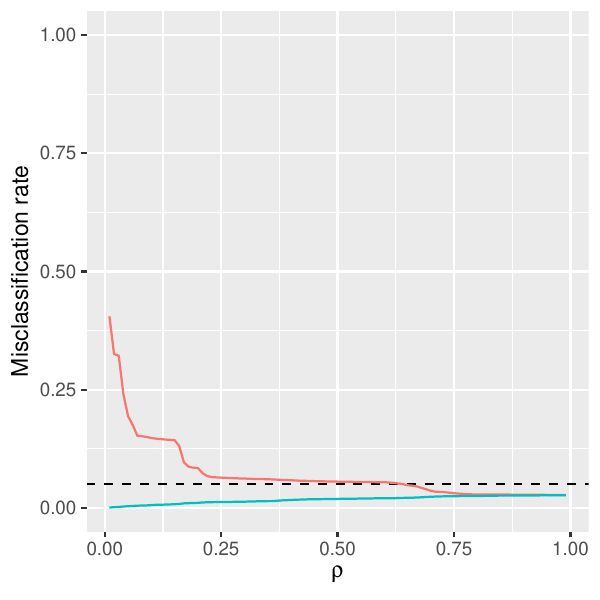} \\ [\abovecaptionskip]
     \small (a) Uniform prior. & 
     \small (b) Proportional prior.
  \end{tabular}
  \caption{The graphs show the misclassification rate $1-R_\rho^{cv}$ as a function of $\rho$ for $\tilde{R}$ in blue and $\dot{R}$ in red. Notice how the two curves are both monotone in $\rho$ (but increasing and decreasing respectively) and converge at $\rho=1$, i.e. when both classifiers output sets of cardinality 1. The dashed black line is the chosen limit $\delta$ on the misclassification rate of 5\%. For (a), the largest value of $\rho$ that does not exceed a misclassification rate of 5\% (for reward function $\tilde{R}$) is $\rho_{0.05}=0.3$, at which 32.9\% of the classified birds  are wrong and/or contain two or more categories. For (b) the corresponding value of is $\rho_{0.05}=1$, and the misclassification rate is the same for both reward functions.}
  \label{misclassificationrates}
\end{figure}

\begin{table}[ht] 
\centering
\begin{tabular}{rrr}
  \hline
 Species & $\rho = 1$ & $\rho = 0.3$ \\ 
  \hline
Reed warbler & 17.4\% & 5.1\%  \\ 
  Blyth's reed warbler & 2.8\%  & 0.9\%  \\ 
  Paddyfield warbler & 12.9\% & 3.2\% \\ 
  Marsh warbler & 6.9\% &  2.2\% \\ 
   \hline
\end{tabular}
\caption{The misclassification (as in rejecting the correct category) rate for each category using the uniform prior and reward function $\tilde{R}$. The majority of observations are partially observed leading to the high rates of misclassification (cf. Section \ref{classification}).}
\label{spmcr}
\end{table}

% In the case of a known cost of misclassifications, and a known cost of 
% having a large indecisive region, one could certainly compute which 
% $\rho$ to use, in order to get the minimal expected cost. However, applications in species identification have more vague costs. The user might not be able to quantify cost, but believes that it is worse to misclassify, i.e. predict a category to which the obervation does not belong, than to not be precise, i.e. only rule out 
% categories with sufficiently low degrees of belief. Together with the choice of 
% prior category probabilities $(\pi_1,\ldots, \pi_N)$ a user may 
% completely accomodate the prior beliefs about the classification 
% problem at hand. Indeed, $(\pi_1,\ldots, \pi_N)$ captures the 
% apriori beliefs about how expected an observation from each category 
% is relative to the others and $\rho$ is intended to represent the 
% user's idea of the cost of misclassification. Last, $\tau$ 
% represents how much outlyingness we accept without losing trust in our classifier.

The parameter $\rho$ can be used as a measure of how managable a classification problem is, given $\delta$. For any classification problem, we may first choose the misclassification rate upper limit $\delta$, and then find the highest value $\rho_\delta$ of $\rho$ for which the misclassification rate, based on reward function $\tilde{R}$, is below $\delta$. A more managable problem would then be one for which $\rho_\delta$ is higher. Classification problems where $\rho_\delta=1$ can be considered automizable, as there are sufficiently few dubious cases not to warrant any closer examination by an expert. 
%For model selection, $\rho$ is useful too; for any combination of traits and/or covariates, we compute the highest $\rho$-value for which the misclassification rate is below $\delta$, and then choose a model where $\rho$ is acceptably high, and the model is as parsimonious as possible.

\section{Comparison with discriminant analysis in the \textsf{R} package \texttt{mclust}} \label{comparison}
The \texttt{mclust}-package is the most popular package for multivariate clustering and classification \citep{mclust} using \textsf{R}. It contains methods for clustering and classification of observations stemming from mixtures of Gaussian distributions, with cross-validation readily available to assess classification accuracy of the discriminant analysis models.
 We will use the discriminant analysis function \texttt{MclustDA} in order to compare its classification performance with our proposed method, using the \textit{Acrocephalus}-data. 
 
If one specifies prior distributions in the call to \texttt{MclustDA}, inference of model parameters is done through maximum aposteriori estimation \citep{fraley2007bayesian}. Then, in the next step these estimates are plugged into a Bayseian classifier with known parameters \citep[Section 6.1 of][]{fraley2002model}. In contrast, our classification scheme is fully Bayesian \eqref{omega}-\eqref{predprob}, with integration over posterior distributions of parameters. The priors on the model parameters differ somewhat, since we assume independent priors on the mean vector and covariance matrix of each class' Gaussian distribution, whereas \texttt{MclustDA} does not \citep{fraley2007bayesian}. Moreover, we may specify priors for the mean vector for each class in our proposed method, while \texttt{MclustDA} does not permit this. Therefore, we chose the default option in \texttt{MclustDA}, which uses the mean of the traits from the data set. For our proposed method, we use the informative priors of Section \ref{acromod}. The prior on the species is uniform in the cross-validation scheme for both methods.
 
A further difference is that \texttt{MclustDA} do not permit the use of partial observations and covariates. Therefore, we reduce the \textit{Acrocephalus}-data considerably to only consist of complete observations and partition it according to the binary covariate \textit{age}. Consequently, two models are fitted with each method. These are then evaluated using 100-fold cross validation to estimate the prediction accuracy.
 
 %  As mentioned in \citet{mcparland2017clustering}, extending Bayesian multivariate mixed response classification models to incorporate covariate data is a possible improvement, which we provide it in this paper. Thus, comparative methods for the type of classification we want to do have some notable differences. We chose to use the \texttt{mclust}-package to make a comparison, although it does not permit covariate usage. However, since we have one binary covariate in our example analysis, we may partition the data according to this covariate and fit classification models to each subset, allowing for comparisons of prediciton errors.
 
Results of the comparison of classification accuracy are presented in detail in Tables \ref{prederr}, \ref{juvclasstab} and \ref{adclasstab}. The take-away is that the methods preform similarily, except for juvenile birds, where our method to a larger extent than \texttt{MclustDA} classifies \textit{Reed warbler} as \textit{Marsh warbler}, and this is the source of its larger overall classification error probability among the juveniles.
 
Lastly, note also that the error probabilities with our method are higher than the estimated error probabilities \eqref{prederr1} of Section \ref{oneclass}. Although the estimates are computed in different ways, this nevertheless indicates the added value of including partial observations (as in Section \ref{oneclass}) when available.
 
 \begin{table}[ht] 
 \caption{Misclassification rate is the probability of classifying a new observation wrongly, estimated by 100-fold cross validation, when new observations are assumed to have a uniform distribution over species.}
\begin{center}
  \begin{tabular}{r|ll}
    Misclassification rate & \multicolumn{2}{c}{Method} \\
    \cline{2-3}
    Age class & \texttt{mclust} & K-H \\
    \hline
%      Paddyfield warbler excluded & Adult & &  \\
%      Paddyfield warbler excluded & Adult & &  \\
     Adult & 3.05 \% & 3.56 \% \\
     Juvenile & 1.59 \% & 3.74 \%
  \end{tabular}
\end{center}
\label{prederr}
 \end{table}

 % latex table generated in R 4.0.0 by xtable 1.8-4 package
% Mon Feb 15 18:55:05 2021
\begin{table}[ht]
\centering
\caption{Classification table for \textit{adult} birds using the \texttt{cvMclustDA} cross-validation method (referred to as \texttt{mclust} below), and the method proposed in this paper, (referred to as K-H below). The similarities of the methods are high, with a slightly higher accuracy for the \texttt{mclust}-method for Blyth's reed warbler, and a slightly higher accuracy for the proposed method for the Paddyfield warbler.}
{\small
\begin{tabular}{r|rr|rr|rr|rr}
 & \multicolumn{8}{c}{True species}\\
 \cline{2-9}
  & \multicolumn{2}{r}{Reed warbler} & \multicolumn{2}{r}{Blyth's reed w.} & \multicolumn{2}{r}{Paddyfield w.} & \multicolumn{2}{r}{Marsh warbler} \\ 
 Predicted species & \texttt{mclust} & K-H & \texttt{mclust} & K-H & \texttt{mclust} & K-H & \texttt{mclust} & K-H\\
   \hline
Reed warbler&  76 & 76 & 0 & 0 & 0 &   0 & 0  &   0 \\ 
  Blyth's reed warbler & 0 & 1 & 65 & 63 & 2 &  1 & 0 & 0 \\ 
  Paddyfield warbler & 0 & 0 & 3 &  5 & 10 &  11 &0&   0 \\ 
  Marsh warbler & 1 &  0 & 0 &  0 & 0 &  0 & 40 & 40 
%   Error prob. & 1.30\% & 1.30 \% & 4.41 \% & 
\end{tabular}
}
\label{juvclasstab}
\end{table}

\begin{table}[ht]
\centering
\caption{Classification table for \textit{juveline} birds using the \texttt{cvMclustDA} cross-validation method (referred to as \texttt{mclust} below), and the method proposed in this paper, (referred to as K-H below). The biggest difference between the two methods is in the lacking ability of K-H to classify Reed warblers correctly. As for adult birds, the proposed method has better accuracy for Paddyfield warblers, and contrary to adult birds, also for Blyth's reed warbler, although the differences are small.}
{\small
\begin{tabular}{r|rr|rr|rr|rr}
 & \multicolumn{8}{c}{True species}\\
 \cline{2-9}
  & \multicolumn{2}{r}{Reed warbler} & \multicolumn{2}{r}{Blyth's reed w.} & \multicolumn{2}{r}{Paddyfield w.} & \multicolumn{2}{r}{Marsh warbler} \\ 
 Predicted species & \texttt{mclust} & K-H & \texttt{mclust} & K-H & \texttt{mclust} & K-H & \texttt{mclust} & K-H\\
   \hline
Reed warbler&  403 & 382 & 0 & 0 & 0 &   0 & 5  &   5 \\ 
  Blyth's reed warbler & 0 & 2 & 40 & 41 & 2 &  1 & 0 & 0 \\ 
  Paddyfield warbler & 0 & 0 & 1 & 0 & 16 &  17 & 0 &   0 \\ 
  Marsh warbler & 6 &  25 & 0 &  0 & 0 &  0 & 409 & 409 
\end{tabular}
}
\label{adclasstab}
\end{table}

\section{Discussion} \label{discussion}

Throughout this paper, we have defined and analysed a classification 
problem with two sets of observations 
for each subject; its trait vector $Y$ and its covariate vector $x$,
where $x$ is informative about the interpretation of $Y$. Since the trait values are often subject to various types of obfuscation, we set up a unified Bayesian framework for these situations, using a latent multivariate 
Gaussian distribution, with parameters estimated through supervised learning and a blockwise Gibbs sampling algorithm. To formalize the classification, we introduced reward functions with a set-valued input argument and two tuning parameters $\rho \in [0,1]$ and $\tau \in [0,1)$. The choice of $\rho$ affects the size and location of the indecisive 
region $\Lambda$ of our discriminant rule with a partial reject option. This region is the part of observational space where our 
classifer does not have sufficient information to rule out all but one category, whereas $\tau$ puts a limit on how much we allow an observation to deviate from the bulk of the data and still allowing it to be classified by our decision rule. Finally, we present a method of covariate and/or trait selection, through cross-validation, in order to obtain classifers that are both parsimonious and efficient. 

Overall, there are two main usages of the method presented in this 
paper. First, one may derive distinguishing characteristics (cf. Appendix \ref{graphs}) of the categories considered, due to interpretability. Secondly, one may use a fitted model to classify new observations with statistical rigour. An example of the 
usefulness of the first case would be an ornithologist with a set 
of birds of known taxa, who doesn't know what morphologically separates
these taxa. Using this method, she may extract for which trait 
measurements there is a high probability of certain taxa and 
thereby create (and write down) an identification key. Further, 
if there are too many combinations of trait levels to memorize, 
the Bayesian procedure we have described may perform the classification in an automized way.   

A number of extensions of our work are possible. The first one regards the distribution of the latent variable. Although a Gaussian distribution is frequently used within quantitative genetics, there are instances when other distributions are preferable. As trait distribution, due to the central limit theorem, the Gaussian is particularly suitable when the traits are assumed to be polygenic and possibly affected by environmental factors \citep{LyWa98}. A conceptually critical but often numerically negligable adjustment is to correct all latent Gaussian distributions by truncating them to positive values for some of the traits. The trait \textit{wing length} in our first real world example has to be positive by definition, and hence we should adjust our Gaussian distribution 
accordingly. However, considering the values of the parameter estimates (see Table \ref{acroestimates}), it would essentially make no difference to impose such a restriction for this particular data set. In other cases, it could be more important. When multimodality and skewness is a concern, an extension to hidden variables with a normal mixture distribution is an option \citep{HaTi96, mcparland2014clustering}. 

A second extension is to incluce nominal (unordered) traits. This can be achieved by incorporating elements of a multinomial probit model \citep{geweke1994alternative, mcparland2014clustering}. Although bird traits are typically ordered, it might be of interest to incorporate nominal traits when more diverse taxa are classified, or when genetic data in terms of unordered genotypes is part of the trait vector. 

Third, it would be of interest to consider model selection methods other than cross-validation. For instance, a Monte Carlo approximation of the BIC criterion, the so called BIC-MCMC approach, performed well in some previous analyses \citep{fruhwirth2011label, mcparland2017clustering}.

Fourth, the reliance on training data with known categories can potentially 
be relaxed, or at least partially so. Without a reference material, which serves as training data, the nature of our problem would correspond to unsupervised learning and to a large extent fall under general model-based clustering \citep{fraley2002model, mcparland2014clustering, mcparland2017clustering} of observations taken from a mixture of multivariate Gaussian distributions. To take a step towards unsupervised learning in our setting, one would need to allow for indecisive regions where classification is ambigous. Such a clustering algorithm with partial reject options would transfer ideas in \citet{chow1970optimum, ripley1996pattern, herbei2006classification} and this paper of having incomplete decisions, to a framework of unsupervised learning. A challenge in this context is to incorporate the effect of covariates. Specifying the number of clusters in advance
would make the problem more well behaved, but might oftentimes be impossible, since the fundamental problem to solve in practice would be to determine this number of clusters. Still, this would allow the method to be used in situations where it is not known how to classify observations at all, and thus make it possible to investigate how many clusters a given data set supports. 

Fifth, modifying the method slightly in order to handle repeated 
measurements is a straightforward task within our multivariate Gaussian framework. The benefit with repeated measurements of the traits is a better understanding of the magnitude of the measurement error, when trait vectors of observations are replaced by averaged trait vectors, for all repeated measurements. One could then include the number of measurements into the classification method, with direct effects on the size of the indecisive region and the accuracy of the classifier.

Sixth, as mentioned in Section \ref{modelformulation}, we assume prior independence between the regression parameters of different categories. This allows 
the effect of a covariate to vary between the categories, as 
opposed to forcing apriori a covariate to have the same effect across categories. 
However, Appendix \ref{hpdint} lists the posterior means of the 
covariate effects from our example model of Section \ref{acromod}, and one may notice 
that the effect is similar for some traits across categories, and 
to some extent even across traits. This indicates a general effect of our covariate, and hence we could 
construct a model that emphasizes such a general effect, by introducing apriori dependencies between the regression parameters of different species.

Finally, it is of interest to apply and extend our discrminant analysis method with partial reject options in order to analyze data sets where some observations are known to belong to several clusters \citep{latouche2011overlapping}. This requires an extended type of reward function, $R(\cI,\rmI)$, where both the classification $\cI$ and the true $\rmI$ are sets of categories.

\section*{Acknowledgements}
The authors are grateful for the data on \textit{Acrocephalus} warblers and Chiffchaffs
provided by Falsterbo Bird Observatory, in particular Lennart Karlsson, 
Bj\"orn Malmhagen and G\"oran Walinder, and for helpful methodological feedback from Felix Wahl and Mathias Lindholm. We thank the Museum of Natural History in Stockholm, Copenhagen and Tring. 

\appendix

\section{Model formulation, complete data} \label{ideal}
Appendices \ref{ideal}-\ref{graphs} contain further mathematical details about the Bayesian model and classification procedure defined in Sections \ref{modelformulation} and \ref{classification} respectively of the main article. In order to make the text self-contained, some details of the main article have been repeated. This Appendix \ref{ideal} contains, in full detail, the derivation of estimators and the posterior category weights for a model using perfectly observed data. Posterior distributions for Bayesian multivariate linear regression with homoscedasticity assumption is readily available in \citet[Sections 2.8 and 2.11]{rossi2012}, and we extend this to allow for heteroscedasticity.

Suppose we have $N$ different categories, contained in the set $\sfN = \{1,\ldots,N\}$, 
with prior probabilities $\pi = (\pi_1,\ldots,\pi_N)$. With full data we measure $q$ 
traits and $p$ covariates of each subject. Let $Y_{ijk}$ be the measurement of trait 
$k$ for subject $j$ in category $i$, where $1\leq i\leq N$, $1\leq j\leq n_i$, $1\leq k\leq q$ 
and $n_i$ is the number of subjects in category $i$. We assume that 
\begin{equation*}
Y_{ij} = (Y_{ij1},\ldots,Y_{ijq}) \sim \text{N}\left(m_{ij}, \bSigma_{ij}\right)
\end{equation*}
are independent random vectors having a multivariate normal distribution, with
\begin{equation*}
m_{ij} = (m_{ij1},\ldots,m_{ijq}) \qquad \text{and} \qquad \bSigma_{ij} = \left(\Sigma_{ijkl}\right)_{k,l=1}^q
\end{equation*}
being the mean vector and the covariance matrix of subject $j$ of category $i$. Let also
\begin{equation*}
x_{ij} = \left(1, x_{ij1},\ldots,x_{ijp}\right) = \left(x_{ijm}\right)_{m=0}^p
\end{equation*}
be the covariate vector of subject $j$ of category $i$. Trait vectors and covariate 
vectors of category $i$ are rows in the matrices $\mY_i = \left(Y_{i1}^\top,\ldots,Y_{in_i}^\top\right)^\top$ 
and $\mX_i= \left(x_{i1}^\top,\ldots,x_{in_i}^\top\right)^\top$ respectively. 
We now proceed by formulating a multivariate and multiple regression model 
\begin{equation} \label{matmodapp}
\mY_i = \mX_i\mB_i + \mE_i
\end{equation}
for category $i$, where $\mB_i=\left(B_{imk}; m=0,\ldots,p; k=1,\ldots,q\right)$ is 
the regression parameter matrix, whose first row consists of intercepts for 
the $q$ traits, $m_{ij}$ is the $j^{\text{th}}$ row of $\mX_i\mB_i$, and 
$\mE_i = \left(E_{i1}^\top,\ldots,E_{in_i}^\top\right)^\top$ is an error 
term matrix with independent rows $E_{ij} \sim\text{N}(0,\mathbf{\bSigma}_{ij})$. 

For use in the construction of a joint prior, and later the derivation 
of the marginal posterior distributions of the parameters, the vectorized 
form of our regression model is needed. Denote the operation of appending
columns of a matrix by $\text{vec}(\cdot)$ (we will also use the 
inverse operation $\text{vec}^{-1}(\cdot)$ on column vectors) and 
rewrite \eqref{matmodapp} as 
\begin{equation} \label{Uapp}
\mU_i = \text{vec}(\mY_i) = \mZ_i \beta_i + \text{vec}(\mE_i)
\end{equation}
with $\beta_i = \text{vec}(\mB_i)$. Denoting an identity matrix of rank $q$ 
with $\mI_q$ and using the matrix tensor product $\otimes$,
\begin{equation} \label{Zapp}
\mZ_i =  \mI_q \otimes \mX_i = \begin{pmatrix} \mX_i & 0 & \cdots & 0 \\
                        0 & \mX_i & \ddots & \vdots \\
                        \vdots & \ddots & \ddots & 0 \\
                        0 & \cdots & 0 & \mX_i \end{pmatrix}
\end{equation}
is a block-diagonal matrix with $q$ blocks along the diagonal.

Now suppose we have $A$ covariance classes $\alpha = 1,\ldots,A$ 
for category $i$ such that
\begin{equation} \label{covmatsapp}
\bSigma_{ij} = \bSigma^\alpha_i \quad \text{if } x_{ij} \in \cX^\alpha,
\end{equation}
where $\cX = \cX^1 \cup \ldots \cup \cX^A$ is a disjoint 
decomposition of the predictor space $\cX$. Assuming a prior $\text{N}\left(\left(b_{i0k},\ldots,b_{ipk}\right)^\top = b_{ik}, \bSigma_{\mB_i}\right)$ on each of the columns $k=1,\ldots,q$ of $\mB_i$, we obtain a prior 
$\text{N}\left(\left(b_{i1}^\top,\ldots,b_{iq}^\top\right)^\top = \beta_{i0}, \mI_q \otimes \bSigma_{\mB_i} = \bSigma_{\beta_i}\right)$ on $\beta_i$. 
Further, assuming prior independence of $\beta_i, \bSigma_i^1, \ldots,\bSigma_i^A$ and imposing an Inverse-Wishart 
distribution $\bSigma_i^\alpha \sim IW(\nu_0, \mathbf{V}_0)$ on the 
covariance matrices in \eqref{covmatsapp} for $\alpha = 1,\ldots,A$, we 
get the joint prior
\begin{equation} \label{jointpriorapp}
p(\beta_i, \bSigma_i^1, \ldots, \bSigma_i^A) = p(\beta_i)\prod_{\alpha=1}^A p(\bSigma_i^\alpha)
\end{equation}
for the parameters of category $i$.

\subsection{Estimation}

Let $\theta_i = \left(\mB_i, \bSigma^1_i,\ldots,\bSigma_i^A\right)$ represent 
all parameters of category $i$. In the following, we assume that 
$\theta_1, \ldots, \theta_N$ are independent random vectors with 
probability densities $p(\theta_1),\ldots,p(\theta_N)$ defined in 
\eqref{jointpriorapp}. Introducing dependencies is of course possible, 
and may be important for specific problems. This is briefly mentioned in Section \ref{discussion} of the main paper. From Bayes' Theorem we get an aposteriori density
\begin{align*}
p(\theta_i\mid \cD_i) &= p(\theta_i){\cal L} \left( {\cal D}_i \right)^{-1} \prod_{j=1}^{n_i} f\left(y_{ij}; x_{ij}, \theta_i\right)\\
&= p(\theta_i){\cal L} \left( {\cal D}_i \right)^{-1}\cL(\theta_i; \cD_i) \\
&\propto p(\theta_i)\cL(\theta_i; \cD_i)
\end{align*}
of $\theta_i$ given the complete training data set 
$\cD_i = \{(x_{ij}, Y_{ij}); \, j=1, \ldots ,n_i\} = \{\mX_i, \mY_i\}$ 
for category $i$. The function 
$\cL\left(\theta_i; \cD_i\right) = p\left(\mY_i \mid \mX_i, \theta_i\right)$ 
is the likelihood, whereas ${\cal L}({\cal D}_i)=p({\cal D}_i)$ is the marginal likelihood of category $i$. In the last step we removed the normalizing factor 
${\cal L}({\cal D}_i)^{-1}$, since it does not depend on $\theta_i$. The Maximum Aposteriori (MAP)-estimator of $\theta_i$ is
\begin{align*}
\theta^{(\text{MAP})}_i &= \argmax_{\theta_i} p(\theta_i\mid \cD_i) \\
&= \argmax_{\theta_i} p(\theta_i)\cL(\theta_i; \cD_i),
\end{align*}
whereas the Bayes' estimator of $\theta_i$ is 
\begin{align*}
\theta^{(\text{Bayes})}_i &= \E\left[\theta_i\mid \cD_i\right] \\
&= \int \! \theta_i p(\theta_i\mid \cD_i) \, \rmd \theta_i \\
&= {\cal L}\left({\cal D}_i\right)^{-1}\int \! \theta_i p(\theta_i)\cL(\theta_i; \cD_i) \, \rmd \theta_i.
\end{align*}
Finally, given a new observation $\cDnew = (x,Y)$, define the 
posterior probability of the new observation belonging to category $i$ as
\begin{align} \label{piapp}
p_i = \Prob (\rmI=i \mid \cD, \cDnew) = \frac{\pi_i\omega_i}{\pi_1\omega_1 + \ldots + \pi_N\omega_N},
\end{align}
where
\begin{align*}
\omega_i &= \int \! f(Y;x,\theta_i)p(\theta_i \mid \cD_i) \, \mathrm{d}\theta_i \\
&= {\cal L}({\cal D}_i)^{-1} \int \! f(Y;x,\theta_i)p(\theta_i) \mathcal{L}(\theta_i; \cD_i) \, \rmd\theta_i
\end{align*}
are the posterior category weights given $\cDnew$ for all categories, 
before the prior probabilities $\pi_i$ have been taken into account.

\subsection{Monte Carlo Approximations}

It is usually difficult to evaluate the normalizing constants 
${\cal L}({\cal D}_i)^{-1}$ for high-dimensional data sets, and hence also 
$\theta^{(\text{Bayes})}_i$ and $\om_i$. However, it is 
possible to estimate $\theta^{(\text{Bayes})}_i$ and 
$\om_i$ by Monte Carlo simulation, with
\begin{equation} \label{Bayesestapp}
\Bayes_i = \frac{1}{R_i}\sum_{r=1}^{R_i} \theta_{ir}
\end{equation}
and
\begin{equation} \label{MComegaYapp}
\hatom_i = \frac{1}{R_i}\sum_{r=1}^{R_i} f(Y;x,\theta_{ir})
\end{equation}
respectively, if $\theta_{i1},\ldots,\theta_{iR_i}$ are $R_i$ replicates 
drawn from the posterior distribution $p(\theta_i\mid \cD_i)$, 
with $\theta_{ir} = \left(\beta_{ir}, \bSigma_{ir}^1,\ldots,\bSigma_{ir}^A\right)$.

We will generate $\theta_{i1}, \ldots, \theta_{iR_i}$ by blockwise 
Gibbs sampling, and for this we need the conditional posterior 
distributions of $\beta_i$ and $\bSigma_i^\alpha$ for 
$\alpha=1,\ldots,A$. To derive those, we need some additional 
notation. Let $\mZ_i^\alpha$, $\mX_i^\alpha$, $\mY_i^\alpha$ 
and $\mU_i^\alpha$ denote the submatrices of 
$\mZ_i$, $\mX_i$, $\mY_i$ and $\mU_i$ corresponding to 
covariance class $\alpha$, and let $\mI_n$ be the identity matrix of order $n$. Recall also that $\mB_i = \text{vec}^{-1}(\beta_i)$, 
meaning that we know $\mB_i$ from $\beta_i$, 
and vice versa. For simplicity of notation we omit index $i$ in the following proposition:

\begin{prop} \label{prop1}
Denote the parameter vector of a Bayesian multivariate multiple 
regression model with $A$ covariance classes by
$\theta = \left(\beta, \bSigma^1,\ldots,\bSigma^A\right)$, where 
$\beta$ is the regression parameter vector and $\bSigma^1,\ldots,\bSigma^A$ 
are the $A$ covariance matrices. Let the prior of $\theta$ be 
$p(\theta) = p(\beta)\prod_{\alpha=1}^A p(\bSigma^\alpha)$, where 
$\beta \sim \text{N}(\beta_0,\bSigma_\beta)$ and 
$\bSigma^\alpha \sim IW(\nu_0, \mV_0)$ for $\alpha=1,\ldots,A$. 
Then the posterior distribution of $\beta \mid \mU, \mZ, \bSigma^1,\ldots,\bSigma^\alpha$ 
is $\text{N}(\tilde\beta, \tilde\bSigma)$, where
\begin{align*} 
\tilde{\bSigma} &= \left[ \bSigma_{\beta}^{-1} + \sum_{\alpha=1}^A (\bSigma^\alpha)^{-1} \otimes \left(\mX^\alpha\right)^\top \mX^\alpha \right]^{-1}
\end{align*} 
and
\begin{align*} 
\tilde{\beta} &= \tilde{\bSigma} \times \left[\bSigma_{\beta}^{-1}\beta_0 + \sum_{\alpha=1}^A \left(\left(\bSigma^\alpha\right)^{-1} \otimes \left(\mX^{\alpha}\right)^\top\right)\mU^{\alpha}\right].
\end{align*}
\end{prop}
\begin{proof}
By applying Bayes' theorem 
\begin{align*}
p\left( \beta \mid \mU, \mZ, \bSigma^1, \ldots \bSigma^A\right) &\propto \exp \left\{ -\frac{1}{2} (\beta - \beta_0)^\top \bSigma_{\beta}^{-1} (\beta - \beta_0) \right\} \cdot \nonumber \\
&\cdot \prod_{\alpha=1}^A \exp\left\{ -\frac{1}{2} \left(\mU^\alpha - \mZ^\alpha \beta\right)^\top \left(\bSigma^\alpha \otimes \mI_{n^\alpha} \right)^{-1} \left(\mU^\alpha - \mZ^\alpha \beta\right) \right\} \\
&= \exp \left\{ -\frac{1}{2} \beta \mC\beta + \beta \mD \right\}
\end{align*}
where $n^\alpha$ is the number of observations in covariance class $\alpha$,
\begin{align*}
\mC &= \bSigma_{\beta}^{-1} + \sum_{\alpha=1}^A (\mZ^\alpha)^\top (\bSigma^\alpha \otimes \mI_{n^\alpha})^{-1} \mZ^\alpha \\
&= \bSigma_{\beta}^{-1} + \sum_{\alpha=1}^A \left(\mZ^\alpha (\bSigma^\alpha \otimes \mI_{p+1})^{-1}\right)^\top \mZ^\alpha \\
&= \bSigma_\beta^{-1} + \sum_{\alpha=1}^A ((\bSigma^\alpha)^{-1}\otimes \mI_{p+1}) (\mZ^\alpha)^\top\mZ^\alpha\\
&= \bSigma_{\beta}^{-1} + \sum_{\alpha=1}^A \left(\bSigma^\alpha\right)^{-1} \otimes \left(\mX^{\alpha}\right)^\top\mX^{\alpha},
\end{align*}
where in the second step of the last equation we used Lemma \ref{lemma1} below, and 
\begin{align*}
\mD &= \bSigma_{\beta}^{-1}\beta_0 + \sum_{\alpha=1}^A (\mZ^\alpha)^\top (\bSigma^\alpha \otimes \mI_{n^\alpha} )^{-1} \mU^\alpha \\
&= \bSigma_{\beta}^{-1}\beta_0 + \sum_{\alpha=1}^A \left(\left(\bSigma^\alpha\right)^{-1} \otimes \left(\mX^{\alpha}\right)^\top\right)\mU^{\alpha}.
\end{align*}
Consequently,
\begin{equation*}
\beta \mid \mU, \mZ, \bSigma^1,\ldots,\bSigma^A \sim N(\tilde{\beta}, \tilde{\bSigma})
\end{equation*}
where
\begin{align*}
\tilde{\beta} &= \mC^{-1}\mD \\
&= \left[ \bSigma_{\beta}^{-1} + \sum_{\alpha=1}^A (\bSigma^\alpha)^{-1} \otimes \left(\mX^\alpha\right)^\top \mX^\alpha \right]^{-1} \\
&\times \left[\bSigma_{\beta}^{-1}\beta_0 + \sum_{\alpha=1}^A \left(\left(\bSigma^\alpha\right)^{-1} \otimes \left(\mX^{\alpha}\right)^\top\right)\mU^{\alpha}\right]
\end{align*}
and
\begin{equation*}
\tilde{\bSigma} = \mC^{-1} = \left[ \bSigma_{\beta}^{-1} + \sum_{\alpha=1}^A (\bSigma^\alpha)^{-1} \otimes \left(\mX^\alpha\right)^\top \mX^\alpha \right]^{-1}.
\end{equation*}
Notice that $\tilde{\beta}$ is a multivariate version of a weighted average of the prior vector $\beta_0$ and the least squares estimates of $\beta$ obtained from the $A$ covariance classes.
\end{proof}
The following lemma was used in the proof of Proposition \ref{prop1}. It admits a considerable gain in computational speed, when calculating the posterior covariance matrix $\tilde{\bSigma}$.
\begin{lemma} \label{lemma1}
Let $\mZ$ be a block-diagonal $nq \times (p+1)q$ matrix, 
where there are $q$ blocks $\mX$, which are $n \times (p+1)$-matrices, 
along the diagonal. Let $\bSigma$ be a symmetric, 
positive definite $q\times q$-matrix. Then it holds that
\begin{equation*}
\mZ^\top \left(\bSigma \otimes \mI_n\right)^{-1}\mZ = \left(\mZ\left(\bSigma \otimes \mI_{p+1}\right)^{-1}\right)^\top\mZ.
\end{equation*}
\end{lemma}
\begin{proof}
We prove the lemma by iterated use of the mixed-product 
property of the tensor product. Since $\mZ = \left(\mI_q \otimes \mX\right)$, 
the left hand side becomes
\begin{align*}
\left(\mI_q \otimes \mX\right)^\top \left(\bSigma \otimes \mI_n\right)^{-1} \left(\mI_q \otimes \mX\right) &= \left(\mI_q \otimes \mX^\top\right) \left(\bSigma^{-1} \otimes \mI_n\right) \left(\mI_q \otimes \mX\right) \\
&= \left(\mI_q\bSigma^{-1} \otimes \mX^\top\mI_n \right) \left(\mI_q \otimes \mX\right) \\
&= \left(\bSigma^{-1} \otimes \mX^\top \right) \left(\mI_q \otimes \mX\right) \\
&= \bSigma^{-1}\mI_q \otimes \mX^\top\mX \\
&= \bSigma^{-1} \otimes \mX^\top\mX \\
\end{align*}
and the right hand side becomes
\begin{align*}
\left(\left(\mI_q \otimes \mX\right)\left(\bSigma \otimes \mI_{p+1}\right)^{-1}\right)^\top \left(\mI_q \otimes \mX\right) &= \left(\mI_q\bSigma^{-1} \otimes \mX\mI_{p+1}\right)^\top \left(\mI_q \otimes \mX\right) \\
&= \left(\bSigma^{-1} \otimes \mX \right)^\top \left(\mI_q \otimes \mX\right) \\
\{\text{by symmetry of $\bSigma$}\}&= \left(\bSigma^{-1} \otimes \mX^\top \right) \left(\mI_q \otimes \mX\right) \\
&= \bSigma^{-1}\mI_q \otimes \mX^\top\mX \\
&= \bSigma^{-1} \otimes \mX^\top\mX \\
\end{align*}
which proves the lemma already in the third equalities.
\end{proof}
Using the vector form, we may express the 
conditional posterior of the regression parameters
\begin{equation*}
\beta_{i} \mid \mU_i, \left\{\bSigma_{i}^\alpha\right\}_{\alpha=1}^A \sim \text{N}(\tilde{\beta},\tilde{\bSigma})
\end{equation*}
where
\begin{equation*}
\tilde{\bSigma} = \left[ \bSigma_{\beta}^{-1} + \sum_{\alpha=1}^A \left(\bSigma_i^\alpha\right)^{-1} \otimes \left(\left(\mX_i^\alpha\right)^\top\mX_i^\alpha \right) \right]^{-1},
\end{equation*}
and 
\begin{align*}
\tilde{\beta} &= \tilde{\bSigma} \times \left[\bSigma_{\beta}^{-1}\beta_0 + \sum_{\alpha=1}^A \left(\left(\bSigma_i^\alpha\right)^{-1} \otimes \left(\mX_i^\alpha\right)^\top\right)\mU_i^\alpha \right].
\end{align*}
Meanwhile, the conditional posteriors of the covariance matrices are
\begin{align*}
\bSigma_i^\alpha \mid \mB_{i}, \mY_i^\alpha, \mX_i^\alpha &\sim IW(\nu_0 + n_i^\alpha, \mV_0 + \mS^\alpha_i),
\end{align*}
where $n_i^\alpha$ denotes the number of observations in the covariance class $\alpha$ for category $i$ and
\begin{align*}
\mS_i^\alpha = \left(\mY_i^\alpha - \mX_i^\alpha \mB_{i}\right)^\top\left(\mY_i^\alpha - \mX_i^\alpha \mB_{i}\right) %+ \left(\mB_{i} - \mB_0\right)^\top \bSigma_{\mB}\left(\mB_{i} - \mB_0\right).
\end{align*}

Having computed $\hatom_1,\ldots,\hatom_N$, for $\cDnew$, we may compute the Monte Carlo-estimated aposteriori probability of $\cDnew$ being in category $i$ as
\begin{equation*}
\hat{p}_i = \hat{\Prob}(I=i\mid \cD,\cDnew) = \frac{\pi_i\hatom_i}{\pi_1\hatom_1+\ldots+\pi_N\hatom_N},
\end{equation*}
where $\cD = \cD_1 \cup \ldots \cup \cD_N$ is the complete training data set.

\section{Model Formulation, obfuscated data} \label{real}
Overall our setup is the same as in Appendix \ref{ideal}, but we now suppose 
there is only partial information about the complete training data 
set $\cD$. Due to some obfuscation, which could be due to rounding, 
grouping, categorization or lost measurements of some traits, 
we only know that 
\begin{equation*}
Y_{ij} \in \sfS_{ij} = \sfS_{ij1}\times \cdots \times \sfS_{ijq},
\end{equation*}
i.e. the complete trait vector $Y_{ij}$ for subject $j$ of category 
$i$ is contained in a  hyperrectangle $\sfS_{ij}$, whose components 
are given by $\{\sfS_{ijk}\}_{k=1}^q$. The components are sets, ranging 
in possible size from singletons to infinite intervals of 
$\mathbb{R}$, and are given by
\begin{equation*}
\sfS_{ijk} = \begin{cases} Y_{ijk}, & k \notin \sfK_{ij}, \\
\left(c_{ijk},d_{ijk}\right], & k \in \sfK_{ij},
\end{cases}
\end{equation*}
where $\sfK_{ij}= \left\{k; 1\le k \le q; \, Y_{ijk} \text{ obfuscated}\right\}$.
As described in the main article, without loss of generality we assume that $z_{ijk}$, the mid point of $(c_{ijk},d_{ijk}]$ for finite sets, is integer-valued.

We can treat all types of obfuscations uniformly in the following way.
Suppose trait $k$ of subject $j$ of category $i$ is imperfectly 
observed, i.e. $k \in \sfK_{ij}$. Let $g_{k}$ be the number of 
categories of this trait, which we number as $0,1,\ldots,g_{k}-1$.
The observed category is $z_{ijk} \in \left\{0,1,\ldots,g_{k} - 1\right\}$, 
where $g_{k} =2$ for binary data and $g_{k} = \infty$ for count data. 
The corresponding side of $\sfS_{ij}$ is 
\begin{equation*}
\sfS_{ijk} = \begin{cases}
            \left(-\infty, \frac{1}{2}\right], & \text{if } z_{ijk} = 0, \\
            \left(z_{ijk} - \frac{1}{2}, z_{ijk} + \frac{1}{2}\right], & \text{if } 1\le z_{ijk} \le g_{k}-2, \\
            \left(g_{k}-\frac{3}{2}, \infty\right), & \text{if } z_{ijk} = g_{k} - 1.
          \end{cases}
\end{equation*}
We also write 
\begin{equation*}
Z_{ijk} = z(\sfS_{ijk}) = \begin{cases}
            0, & \text{if } \sfS_{ijk} = \left(-\infty, \frac{1}{2}\right], \\
            \frac{c_{ijk} + d_{ijk}}{2}, & \text{if $\sfS_{ijk}$ is bounded}, \\
            g_{k} - 1, & \text{if } \sfS_{ijk} = \left(g_{k}-\frac{3}{2}, \infty\right],
          \end{cases}
\end{equation*}
for the center point of a finite or half-open, infinite $\sfS_{ijk}$, 
whereas $z\left(\sfS_{ijk}\right) = Y_{ijk}$ when $Y_{ijk} = \sfS_{ijk}$ 
is perfectly observed. We may write the actually observed training 
data set as
\begin{equation*}
\cDobs = \left\{\left(x_{ij},\sfS_{ij}\right);\, i=1,\ldots,N,j=1,\ldots,n_i\right\}.
\end{equation*}

\subsection{Estimation}
Using $\cDobs_i = \{(x_{ij},\sfS_{ij}) ; j=1,\ldots,n_i\}$, the 
posterior distribution of $\theta_i$ becomes
\begin{align} \label{censpost}
p(\theta_i \mid  \cDobs_i) &= p(\theta_i)\cL \left( \cDobs_i \right)^{-1}\prod_{j=1}^{n_i} p(\sfS_{ij};x_{ij},\theta_{i}) \nonumber \\
 &= p(\theta_i) \cL \left( \cDobs_i \right)^{-1} \cL(\theta_i ; \cDobs_i),
\end{align}
where the normalizing factor 
\begin{equation*}
\cL\left(\cDobs_i \right) = \int \! p(\theta_i) \cL(\theta_i; \cDobs_i) \, \rmd\theta_i
\end{equation*}
is the marginal likelihood of category $i$, and 
\begin{align} \label{censlik}
p(\sfS_{ij},x_{ij};\theta_i) = \begin{cases} f(Y_{ij};x_{ij},\theta_i), & \sfK_{ij}=\emptyset, \\
\int_{\sfS_{ij}} \! f(y_{ij};x_{ij},\theta_i) \prod_{k\in \sfK_{ij}} \rmd y_{ijk}, & \sfK_{ij}\ne\emptyset.
\end{cases}
\end{align}
Thus, with perfect observations ($\sfK_{ij}=\emptyset$), we 
evaluate the density of the trait vector $f$ at the observed 
point $Y_{ij}$ and the model is exactly as specified in Appendix \ref{ideal}. 
Otherwise, we construct a $\left|\sfK_{ij}\right|$-dimensional 
integral over $f$ and the contribution to the likelihood is this 
integral, with the function evaluated exactly at the remaining perfectly 
observed traits, if any exist. In particular, if all traits are 
imperfectly observed, the integral is $q$-dimensional. We may 
approximate the integral in \eqref{censlik} by
\begin{align*}
\left|\sfS_{ij}\right| f\left(z(\sfS_{ij}), x_{ij}; \theta_i\right) = \prod_{k\in \sfK_{ij}} \left|\sfS_{ijk}\right| \cdot f\left(z(\sfS_{ij1}),\ldots,z(\sfS_{ijq}), x_{ij}; \theta_i\right)
\end{align*}
whenever all $\left|\sfS_{ijk}\right| < \infty$ for $k \in \sfK_{ij}$, 
which is the case when employing the trick with auxiliary categories, mentioned in Section \ref{realmod} of the main article.

Alternatively, since $p(\sfS_{ij},x_{ij};\theta_i)$ potentially contains integrals 
of a multivariate Gaussian density function and there in general 
is a lack of a CDF on closed form for this distribution, the integrals in \eqref{censlik} need to be solved numerically. However, in the case of $|\sfK_{ij}| = 1$, with $\sfK_{ij} = \{k\}$ and $\sfS_{ijk}=\left(c_{ijk}, d_{ijk}\right]$, the integral is univariate and thus\footnote{The notation with subscript $(-k)$ means dropping element $k$ from a vector; dropping row $k$ from a matrix when not being the last index of a matrix; and dropping column $k$ when being the last index.}
\begin{align} \label{unilik}
p(\sfS_{ij},x_{ij};\theta_i) &= f(Y_{ij(-k)}, x_{ij}; \theta_i) \bigg[ \Phi\left(\frac{d_{ijk}-m_{ijk}(y_{ij(-k)})}{\sigma_{ijk}}\right) \\
&- \Phi\left(\frac{c_{ijk}-m_{ijk}(y_{ij(-k)})}{\sigma_{ijk}}\right)\bigg], \nonumber
\end{align}
where $m_{ijk}(y_{ij(-k)})=m_{ijk}+\bSigma_{ijk(-k)}\bSigma_{ij(-k)(-k)}^{-1}(y_{ij(-k)}-m_{ij(-k)})$ 
is the conditional expectation of $Y_{ijk}$ given that 
$Y_{ij(-k)} = (Y_{ijk^\prime};\, k^\prime\ne k) = y_{ij(-k)}$, 
$\sigma_{ijk}=\sqrt{\bSigma_{ijkk}-\bSigma_{ijk(-k)}\bSigma_{ij(-k)(-k)}^{-1}\bSigma_{ij(-k)k}}$ 
is the conditional standard deviation of $Y_{ijk}$ given any value of $Y_{ij(-k)}$, 
and $\Phi$ is the CDF of the univariate Gaussian distribution with mean 0 
and standard deviation 1.

Using $\cDobs_i$, we find from \eqref{censpost} that the estimators $\theta^{(\text{MAP})}_i$ and $\theta^{(\text{Bayes})}_i$ are
\begin{align*}
\theta^{(\text{MAP})}_i &= \argmax_{\theta_i} p(\theta_i \mid  \cDobs_i) \nonumber \\
 &=  \argmax_{\theta_i} p(\theta_i) \cL(\cDobs_i ; \theta_i)
\end{align*}
and
\begin{align} \label{Bayescens}
\theta^{(\text{Bayes})}_i &= \E\left[\theta_i \mid  \cDobs_i\right] = \int \! \theta_i p(\theta_i \mid  \cDobs_i) \, \rmd\theta_i \nonumber \\
 &= \cL \left( \cDobs_i \right)^{-1} \int \! \theta_i p(\theta_i) \cL(\cDobs_i ; \theta_i) \, \rmd\theta_i
\end{align}
respectively. Furthermore, redefining $\cDnew := (x,\sfS)$ for a 
new observation, where $\sfS=\sfS_1\times \ldots \times \sfS_q$, and denoting the corresponding set of imperfectly observed traits by $\sfK$, leads 
to the posterior category weights
\begin{align}
\omega_i &= \iint_{\sfS} \! f(y;x,\theta_i) \prod_{k\in \sfK} \mathrm{d}y_{k} \, p(\theta_i \mid  \cDobs_i) \, \mathrm{d}\theta_i \nonumber \\
 &= \cL \left( \cDobs_i \right)^{-1} \iint_{\sfS} \! f(y;x,\theta_i) \prod_{k\in \sfK} \mathrm{d}y_{k} \, p(\theta_i) \cL (\cDobs_i ; \theta_i) \, \mathrm{d}\theta_i \label{omegaS}
\end{align}
of this observation.

\subsection{Monte Carlo Approximations} \label{realMCapprox}
The integral over $\sfS$ in \eqref{omegaS} is, as mentioned in 
conjunction with \eqref{unilik}, potentially impossible to 
compute analytically, but could also be well behaved. We 
can in theory approximate $\theta^{(\text{Bayes})}_i$ in \eqref{Bayescens} as in 
\eqref{Bayesestapp} by sampling $\theta_i$ from $p(\theta_i \mid  \cDobs)$ 
a total of $R_i$ times. However, this entails a large number of 
numerical evaluations of integrals, see \eqref{censpost}-\eqref{censlik}. 
Similarly, we may estimate $\omega_i$ for $1\le i \le N$ in \eqref{omegaS} through
\begin{equation} \label{omhat}
\hat{\omega}_i = \frac{1}{R_i} \sum_{r=1}^{R_i} \int_{\sfS} f(y;x,\theta_{ir}) \prod_{k\in \sfK} \rmd y_{k},
\end{equation}
which in addition to previously presented integrals, involves 
computation of an integral over $\sfS$. As an alternative way 
of computing \eqref{Bayescens} and \eqref{omhat}, we also 
present an approach where complete data is sampled, based on 
the obfuscated data, as one step of the Monte Carlo algorithm, 
whereas the parameters are sampled as another step of the same algorithm. 
This allows us to estimate all $\theta^{(\text{Bayes})}_i$ and $\omega_i$ 
under widespread obfuscation, given that we are able to 
simulate $\mY_i$, $i=1,\ldots,N$. Overall, we want to generate
\begin{equation} \label{paramsetapp}
\left\{\theta_{ir}, Y_{ijkr}, \, 1\le j\le n_i, \, k\in\sfK_{ij}; Y_{kr}, k\in\sfK\right\}_{r=1}^{R_i}
\end{equation}
from
\begin{equation}\label{sampdens}
p(\theta_i|\cDobs_i)\prod_{j=1}^{n_i} f\left(y_{ij\sfK_{ij}r} \mid x_{ij}, \sfS_{ij},\theta_i\right) f\left(y_{\sfK r}, \mid x, Y_{\sfK^\complement};\theta_i\right)
\end{equation}
where $\theta_{ir} = \left(\beta_{ir}, \bSigma_{ir}^1,\ldots,\bSigma_{ir}^A\right)$, 
$y_{ij\sfK_{ij}r} = \left(y_{ijkr}; k\in \sfK_{ij}\right)$, 
$y_{\sfK r}=\left(y_{kr} ; k\in\sfK\right)$ and $Y_{\sfK^\complement} = \left(Y_k ; k\notin \sfK\right)$. 
Note that we do not condition on $\sfS$ in the conditional density of 
the unobserved traits for the new observation that we want to classify, as this would introduce a bias in the Monte Carlo estimate of $\omega_i$ below.
 
The details of the specific Gibbs sampling approach we use are presented in Appendix \ref{MCapp}. 
Having generated a sample $\theta_{i1},\ldots,\theta_{iR_i}$, we may compute 
the estimated category weights of $\cDnew$ as
\begin{equation} \label{MComega}
\hat{\omega}_i = \frac{1}{R_i} \sum_{r=1}^{R_i} f\left(Y_{\sfK^\complement}; x , \theta_{ir}\right) \I_{\left\{Y_{\sfK r} \in \bigtimes_{k\in\sfK} \sfS_k \right\}},
\end{equation}
where $Y_{\sfK^\complement}$ is as above, and $Y_{\sfK r} = \left(Y_{kr} ; k\in \sfK\right)$. 
For every $\theta_{ir}$, one could generate many $Y_{\sfK}$ and 
replace the indicator with an average of the indicators for each 
sampled $Y_{\sfK}$.

A potentially more efficient method would be to define $\{y_t\}_{t=1}^T$, 
where $y_t= (y_{t1},\ldots,y_{tq})$ with $y_{t\sfK^\complement} = Y_{\sfK^\complement}$ 
and $y_{t\sfK} \in \bigtimes_{k\in\sfK} \sfS_k$, in such a way 
that $\{y_{tk}, k \in \sfK\}$ is a grid approximation of $\bigtimes_{k\in\sfK} \sfS_k$. 
Then we can estimate $\omega_i$ through
\begin{equation} \label{MComega3}
\hat{\omega}_i = \frac{\prod_{k\in\sfK} \lvert \sfS_k \rvert}{TR_i} \sum_{r=1}^{R_i}\sum_{t=1}^T f\left(y_t ; x, \theta_{ir}\right).
\end{equation}
If we use the trick with auxiliary categories described in Section \ref{realmod} of the main article, we can we can choose $y_t$ 
uniformly at random on $\bigtimes_{k\in\sfK} \sfS_k$, as long as we 
do not have any missing observations, since those are represented 
with an infinite interval. Thus, \eqref{MComega3} is potentially 
more effcient than \eqref{MComega}, but comes at a cost of generality, 
since \eqref{MComega} is applicable to any new observation.

Finally, the Monte Carlo-estimated aposteriori probability of $\cDnew = (x,\sfS)$ being of category $i$ is
\begin{equation} \label{predprobapp}
\hat{p}_i = \hat{\Prob}(I=i \mid  \cDnew, \cDobs) = \frac{\pi_i\hat{\omega}_i}{\pi_1\hat{\omega}_1 + \ldots + \pi_N\hat{\omega}_N}
\end{equation}
and we may apply \eqref{omegaS} with replacement of $\omega_i$ by 
\eqref{MComega} for prediction. If \eqref{MComega3} is inserted into 
\eqref{predprobapp} we notice that $\prod_{k\in\sfK} \lvert \sfS_k \rvert$ 
and $T$ cancel out, and in case $R_i = R$ for $i=1,\ldots,N$, also $R$ cancels out.

\section{Gibbs sampling details} \label{MCapp}

The focus of this appendix is Procedure \ref{MCalg2}, in which 
we describe in detail how to generate a sample of size $R_i$ 
from the posterior distribution of the parameter vector 
$\theta_i$, using blockwise Gibbs sampling. It describes the 
general case, i.e. when we have obfuscated trait measurements 
in $\cDobs$. For the case with perfectly observed trait 
measurements, we skip the sampling of $\mY_i$ and use the 
observed values instead, otherwise the procedure is the same. 
Applying the procedure to data from each category $i$ will 
yield all the samples we need from the posterior distribution in order
to perform classification.

In Procedure \ref{MCalg2}, $TN(\mu,\bSigma,\sfS)$ refers to 
the truncated Gaussian distribution, where $\mu$ is the mean 
vector, $\bSigma$ is the covariance matrix and $\sfS$ 
is a hyper rectangle specifying the truncation limits. Simulating from this 
distribution can be done exactly using rejection sampling, 
or approximately using an inner Gibbs algorithm. Depending 
on application, either approach can be preferred, as the 
tradeoff is exact sampling versus efficiency. Also, more 
advanced algorithms such as Importance Sampling-techniques 
can be used in this step.

\begin{algorithm}
\caption{The Monte Carlo approach to sampling the parameters' posterior distribution under obfuscation.}
\label{MCalg2}
\begin{algorithmic}
\REQUIRE{$\cDobs, \nu_0, \mV_0, \beta_{i0} = \text{vec}(\mB_{i0})$}
\ENSURE{A sample of size $R_i$ from the posterior distribution of $\theta_i$.}
\FOR {$\alpha=1 \to A$}
\STATE draw $\bSigma_{i0}^\alpha \sim IW\left(\nu_0, \mV_0\right)$
\ENDFOR
\STATE draw $\beta_{i0} \sim \text{N}\left(\beta_{i0}, \bSigma_{\beta_i}\right)$
\STATE $\theta_{i0} \leftarrow \left(\beta_{i0}, \bSigma_{i0}^1,\ldots,\bSigma_{i0}^A\right)$
\FOR{$r=1 \to R_i$}
\FOR {$j=1 \to n_i$}
\STATE draw $\mY_{ij,r-1} \mid x_{ij}, \sfS_{ij}, \theta_{i,r-1} \sim TN\left(\mX_{ij}\mB_{i,r-1},\bSigma_{ij,r-1}, \sfS_{ij}\right)$
\ENDFOR
\STATE $\mU_{i,r-1} \leftarrow \text{vec}(\mY_{i,r-1})$
\STATE draw $\beta_{ir} \mid \mU_{i(r-1)}, \left\{\bSigma_{i(r-1)}^\alpha\right\}_{\alpha=1}^A \sim \text{N}\left(\tilde{\beta}, \tilde{\bSigma}\right)$
\FOR {$\alpha = 1 \to A$}
\STATE draw $\bSigma_{ir}^\alpha \mid \mU_{i(r-1)}^\alpha, \mX_i^\alpha, \beta_{ir} \sim IW(\nu_0 + n_i^\alpha, \mV_0 + \mS_i^\alpha)$
\ENDFOR
\STATE $\theta_{ir} \leftarrow \left(\beta_{ir}, \bSigma_{ir}^1,\ldots,\bSigma_{ir}^A\right)$
\STATE save $\theta_{ir}$
\ENDFOR
\end{algorithmic}
\end{algorithm}
\clearpage

\begin{landscape}
\section{Highest posterior density intervals of parameters} \label{hpdint}
\setlength\tabcolsep{0.25em}
\begin{table}[H]
\centering
\caption{Quantiles of the posterior distribution for the parameters of the \textit{Acrocephalus} model. The values in the \textit{ad}-columns are the change in intercept for adult birds in the \textit{Regression parameter values}-section, and separate values in the other sections. We have chosen to present the correlation between traits instead of the covariances. Notice the very slim posterior density intervals for \textit{Reed Warbler} and \textit{Marsh Warbler}, since we have lots of observations for these species, and the widely dispersed posterior distributions for Paddyfield Warbler, being a species with very few observations and a vague prior on the covariance matrices.}
% \vspace{5pt}
{\footnotesize
\begin{tabular}{rrrrrrr|rrrrrr|rrrrrr}
& \multicolumn{6}{c}{Regression parameter values} & \multicolumn{6}{c}{Variance estimates} & \multicolumn{6}{c}{Correlation estimates} \\
& \multicolumn{2}{c}{Wing} & \multicolumn{2}{c}{Notch} & \multicolumn{2}{c}{Notch pos.} & \multicolumn{2}{c}{Wing} & \multicolumn{2}{c}{Notch} & \multicolumn{2}{c}{Notch pos.}& \multicolumn{2}{c}{Wing-N.} & \multicolumn{2}{c}{Wing-N.p.} & \multicolumn{2}{c}{N. - N.p.}\\
 Quantile \vline & juv & ad& juv & ad & juv & ad & juv & ad & juv & ad & juv & ad& juv & ad & juv & ad & juv & ad \\
  \hline
 \multirow{3}{*}{Reed warbler}
2.5\% \vline &66.69&0.7&10.99&1.27&107.87&1.95&2.31&2.61&0.57&0.55&1.28&1.3&0.46&0.33&-0.09&0.00&0.46&0.29\\
  50\% \vline &66.7&0.73&11.07&1.38&108&2.29&2.35&2.67&0.63&0.62&1.48&1.59&0.52&0.41&0.02&0.20&0.54&0.45\\
  97.5\% \vline &66.72&0.76&11.13&1.52&108.13&2.56&2.38&2.74&0.71&0.72&1.66&2.35&0.57&0.49&0.13&0.37&0.61&0.60\\
  \hline
\multirow{3}{*}{Marsh warbler}
2.5\% \vline &69.99&0.57&9.42&0.54&104.86&0.39&2.07&2.28&0.41&0.47&0.69&0.5&0.37&0.36&-0.13&-0.12&0.23&0.03\\
  50\% \vline &70.05&0.69&9.45&0.61&104.96&0.66&2.19&2.51&0.44&0.53&0.8&0.76&0.40&0.43&-0.03&0.16&0.33&0.28\\
  97.5\% \vline &70.1&0.81&9.48&0.68&105.04&1&2.32&2.76&0.47&0.6&0.92&1.16&0.44&0.49&0.06&0.41&0.42&0.49\\
  \hline
\multirow{3}{*}{Paddyfield warbler}
2.5\% \vline &56.64&-0.74&12.24&0.4&114.8&0.05&1.28&2.36&0.49&0.76&0.77&0.68&-0.20&-0.08&-0.58&-0.11&-0.13&-0.03\\
  50\% \vline &57.27&0.36&12.66&1.16&115.32&0.81&2.17&4.32&0.83&1.35&1.3&1.24&0.19&0.37&-0.25&0.35&0.27&0.42\\
  97.5\% \vline &57.89&1.47&13.1&1.89&115.83&1.61&4&9.07&1.52&2.77&2.46&2.6&0.53&0.69&0.15&0.68&0.59&0.72\\
  \hline
  \multirow{3}{*}{Blyth's reed warbler}
2.5\% \vline &61.84&-0.48&12.23&0.8&113.23&0.63&1.31&1.34&0.47&0.61&0.6&0.87&0.07&0.29&-0.18&-0.17&-0.09&0.09\\
  50\% \vline &62.26&0.03&12.49&1.13&113.53&1.03&1.94&1.85&0.69&0.83&0.9&1.2&0.35&0.49&0.12&0.07&0.21&0.32\\
  97.5\% \vline &62.69&0.56&12.75&1.46&113.85&1.43&3.04&2.64&1.08&1.17&1.37&1.71&0.58&0.64&0.40&0.30&0.47&0.51\\
\end{tabular}}
\label{tab:quantiles}
\end{table}
\end{landscape}

\section{Visualized decision regions} \label{graphs}
All of these visualizations are done using the same model fit and the same generated new observations from the posterior predictive distribution as in Section \ref{classification} of the main text. As a reminder, we used the values $\rho=0.1$ and $\tau = 0.001$ for the tuning parameters of the classifier.

\begin{figure}[h!]
  \centering
  \begin{tabular}{cc}
    \includegraphics[width=0.5\linewidth]{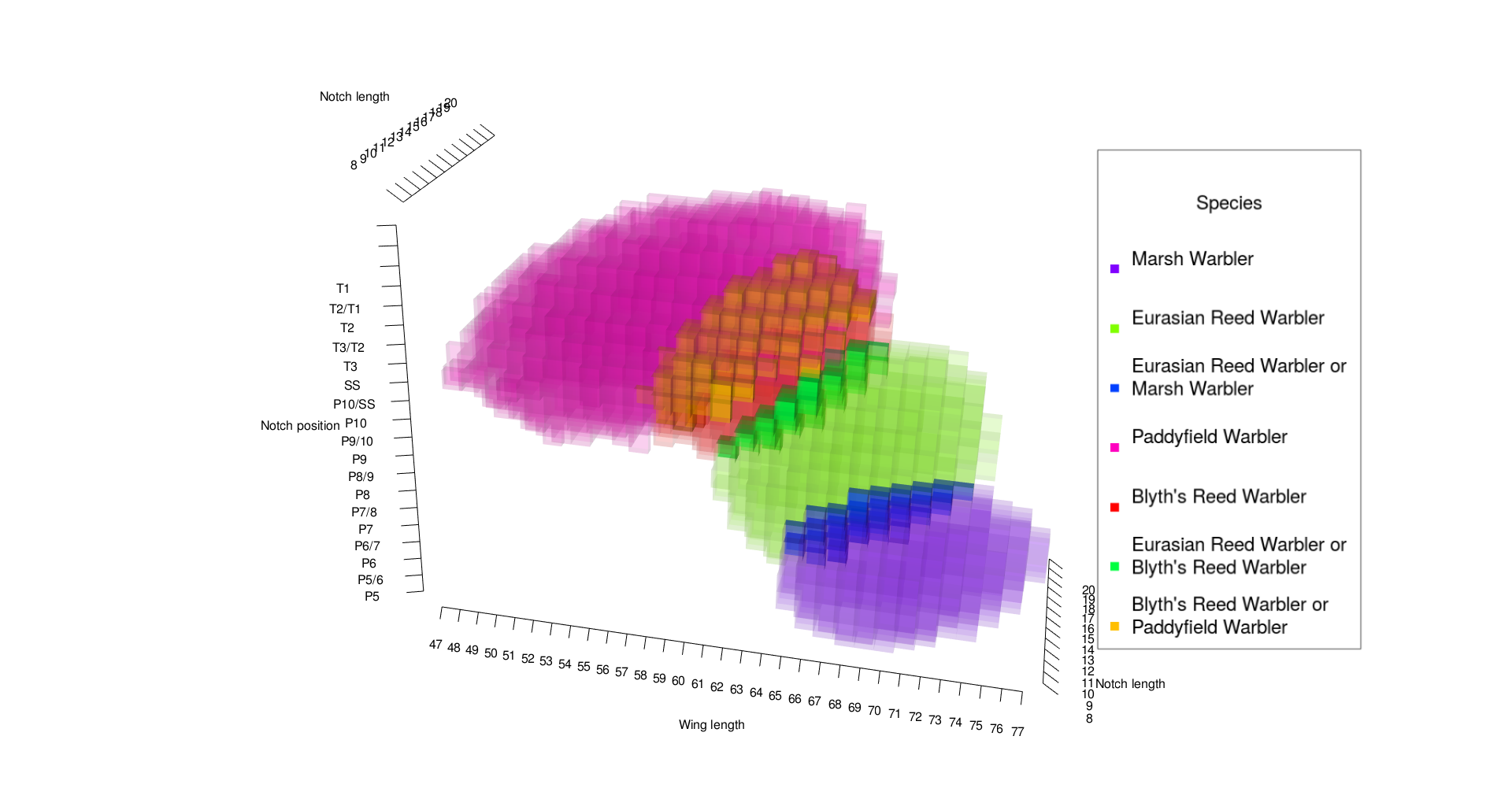} & 
    \includegraphics[width=0.5\linewidth]{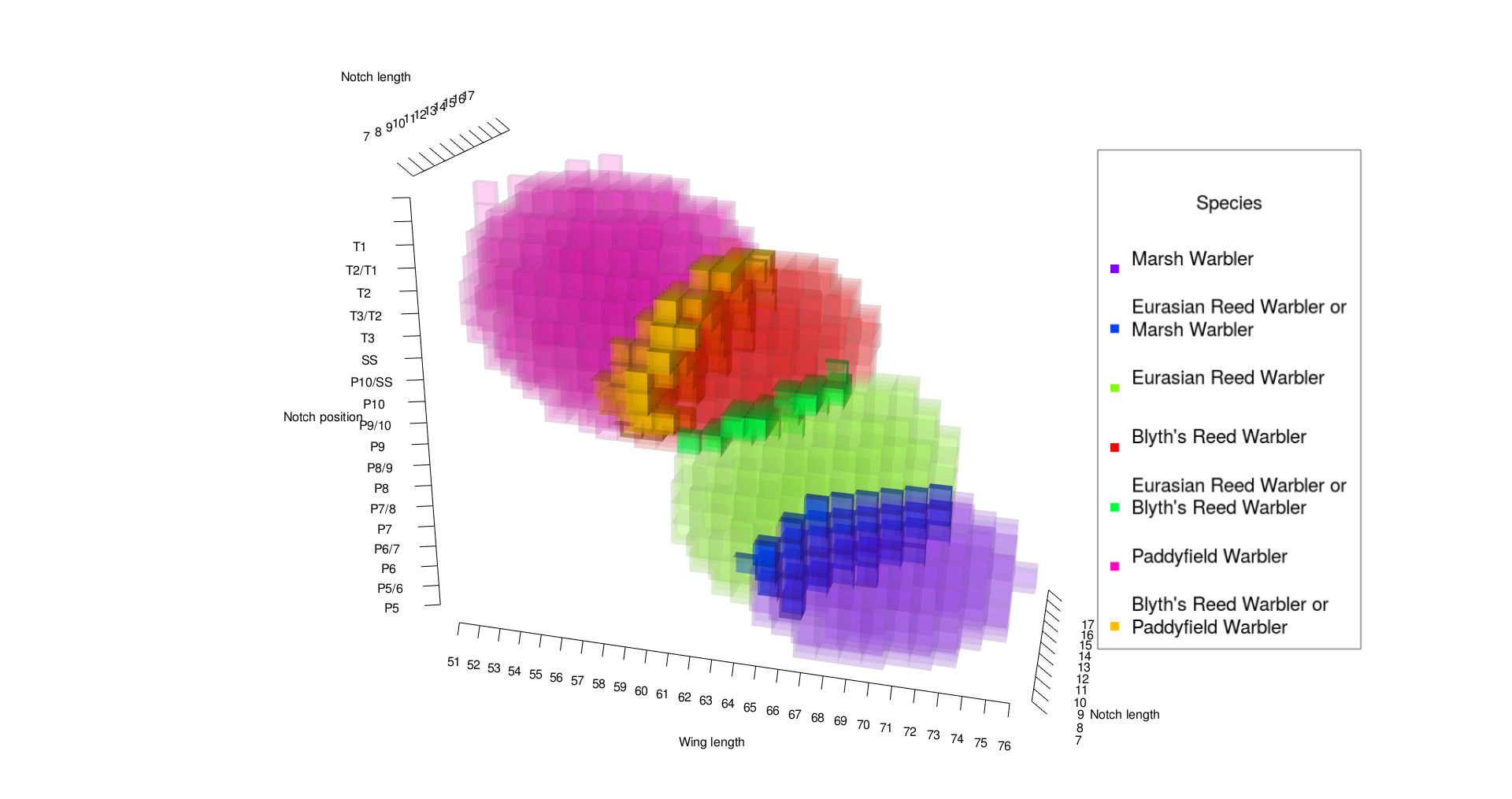} \\ [\abovecaptionskip]
    \small (a) Adult birds. & 
    \small (b) Juvenile birds.
  \end{tabular}
  \caption{Decision regions when observing all three traits of the \textnormal{Acrocephalus} 
  warblers. Completely transparent blocks represent observations that will be classified as outliers, i.e. not get any species assigned to them. The indecisive region 
  $\Lambda$ is less transparent, and colored according to which species there are uncertainty about. The probability of observing an individual that belongs to the indecisive region is $0.0820$ for (a) and $0.0791$ for (b), when each species is equally likely apriori to occur. The 
  decision region of \textnormal{Paddyfield Warbler} partially engulfs \textnormal{Blyth's Reed Warbler}
  for adult birds, reflecting the large uncertainty in the parameter estimates 
  for adult \textnormal{Paddyfield Warblers}. Notice also that we introduce unnamed categories 
  for \textnormal{notch position}, as the predictive posterior distribution requires this.}
  \label{acrocephalus}
\end{figure}

\begin{figure}[h!]
  \centering
  \begin{tabular}{cc}
    \includegraphics[width=0.5\linewidth]{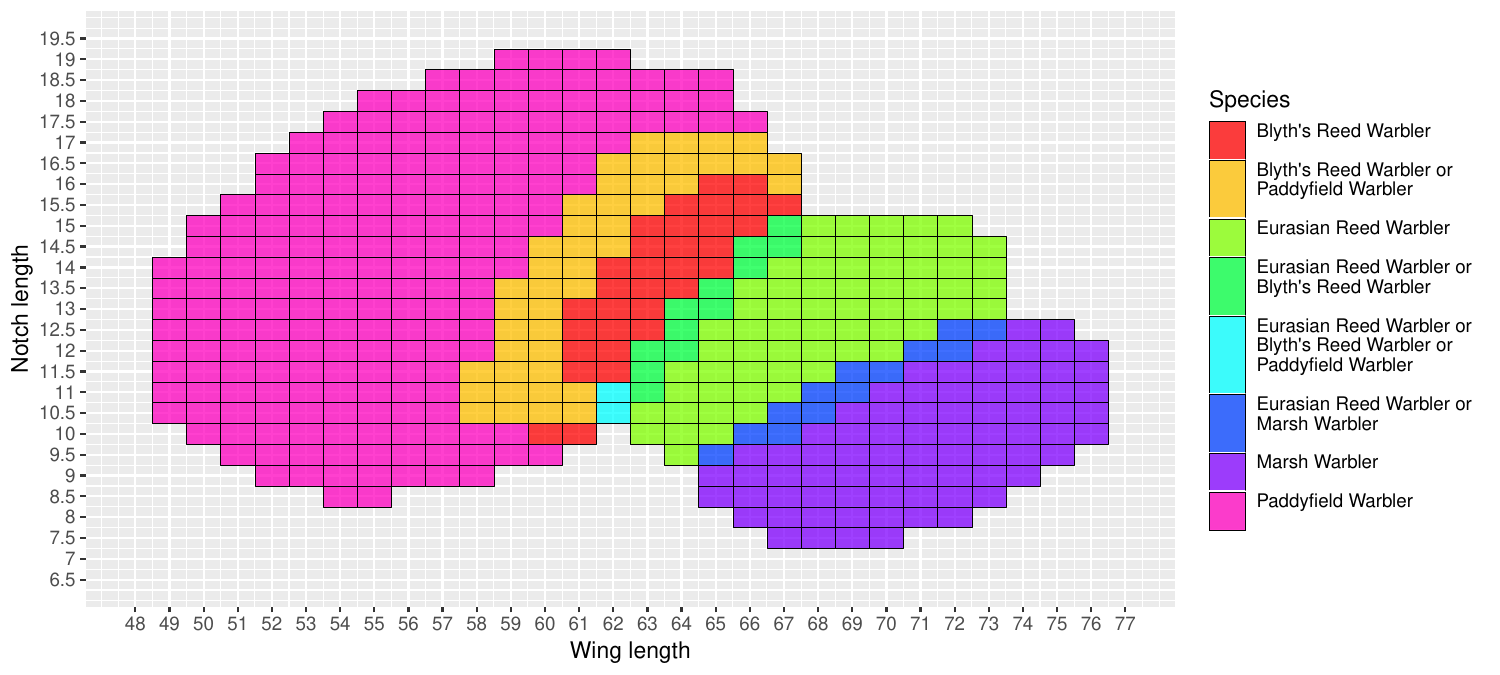} & 
    \includegraphics[width=0.5\linewidth]{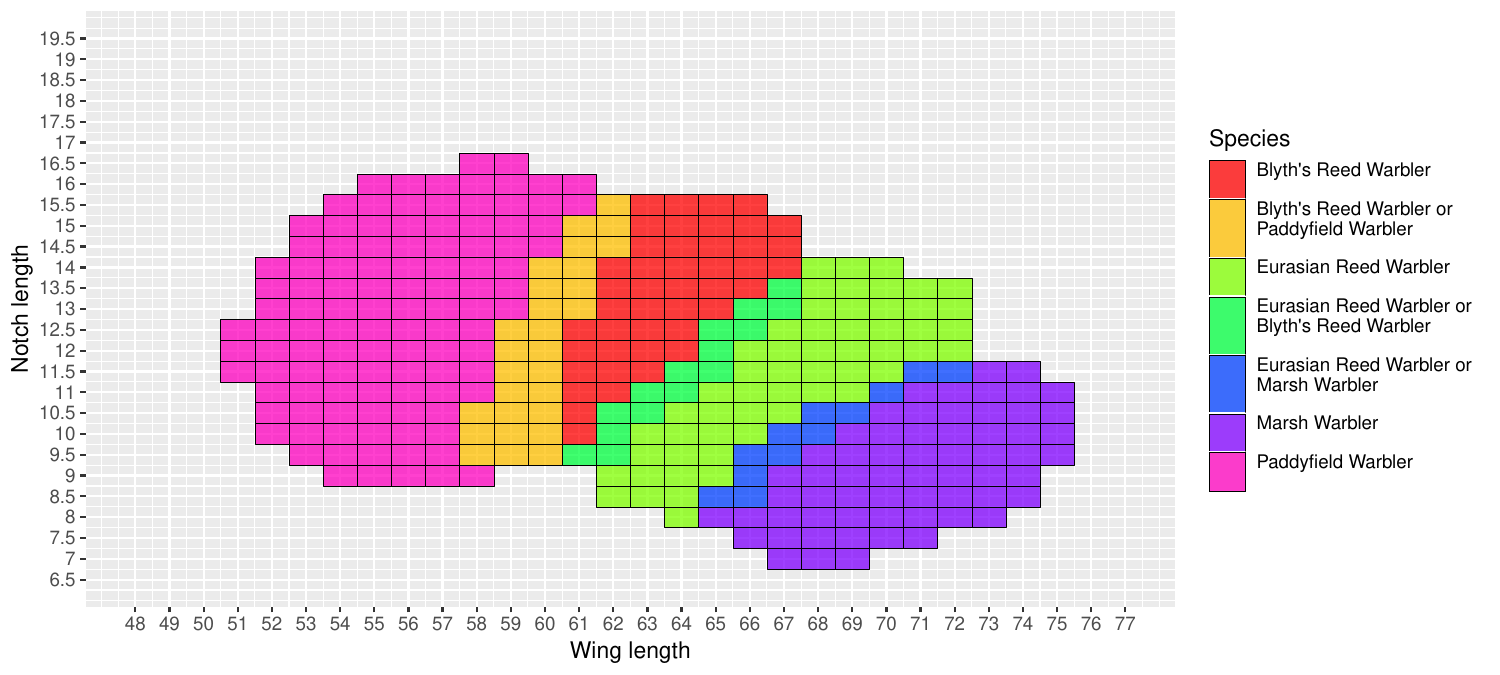} \\ [\abovecaptionskip]
     \small (a) Adult birds. & 
     \small (b) Juvenile birds.
  \end{tabular}
  \caption{Decision regions when only observing wing and notch length.}
\end{figure}

\begin{figure}[h!]
  \centering
  \begin{tabular}{cc}
    \includegraphics[width=0.5\linewidth]{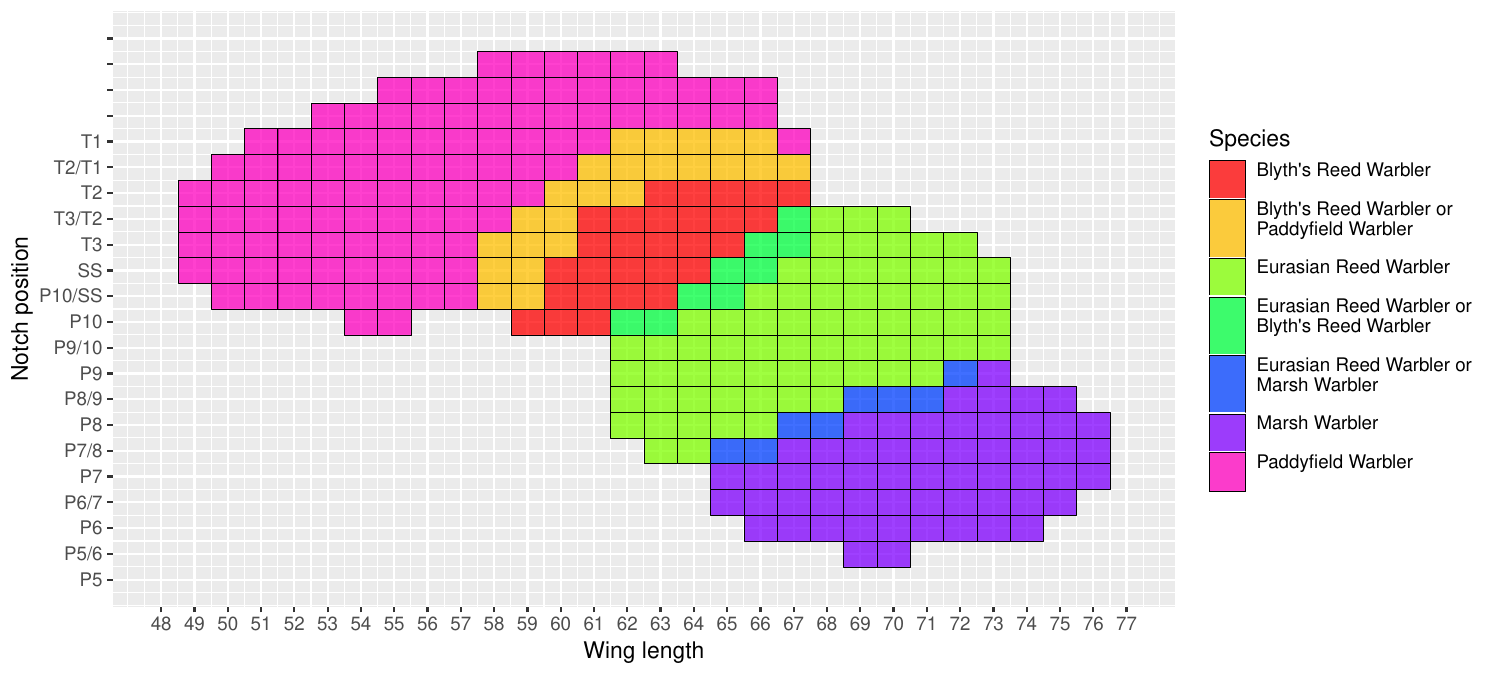} & 
    \includegraphics[width=0.5\linewidth]{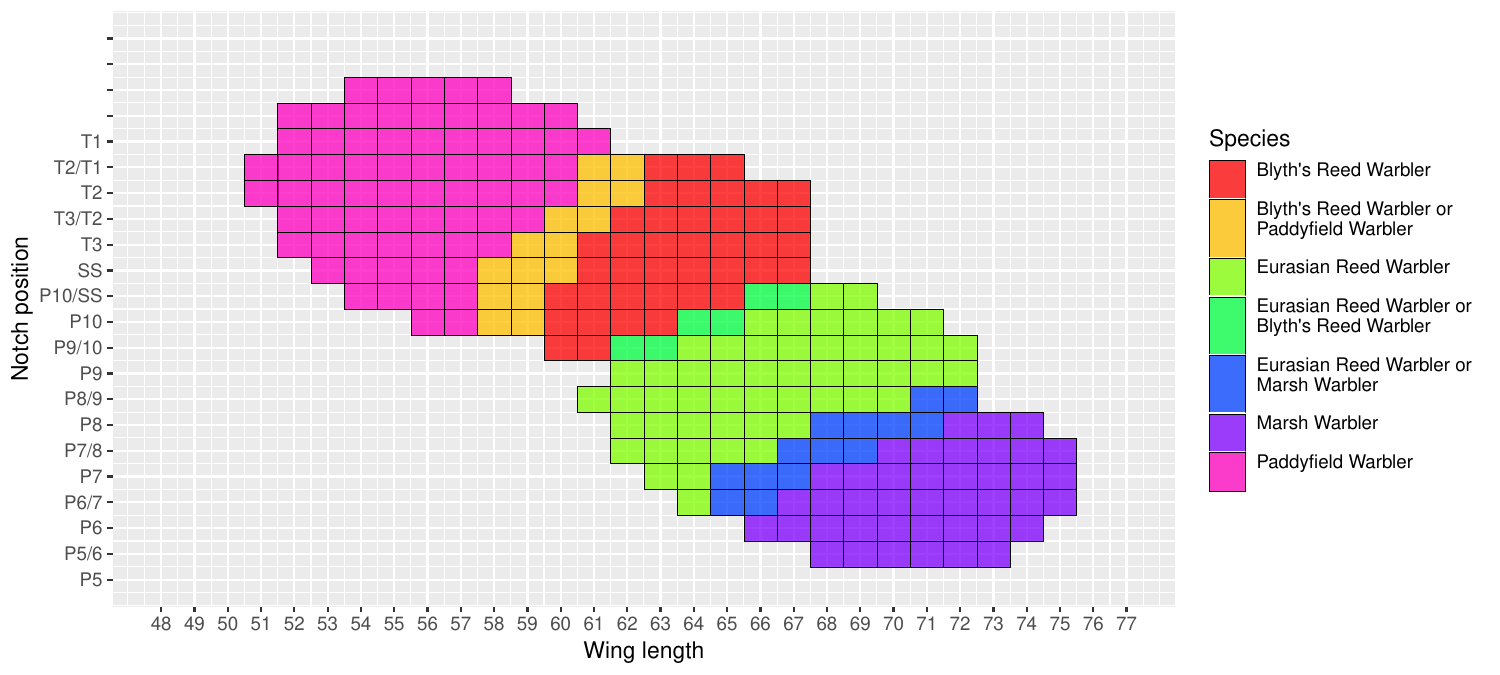} \\ [\abovecaptionskip]
     \small (a) Adult birds. & 
     \small (b) Juvenile birds.
  \end{tabular}
  \caption{Decision regions when only observing wing length and notch position.}
\end{figure}

\begin{figure}[h!]
  \centering
  \begin{tabular}{cc}
    \includegraphics[width=0.5\linewidth]{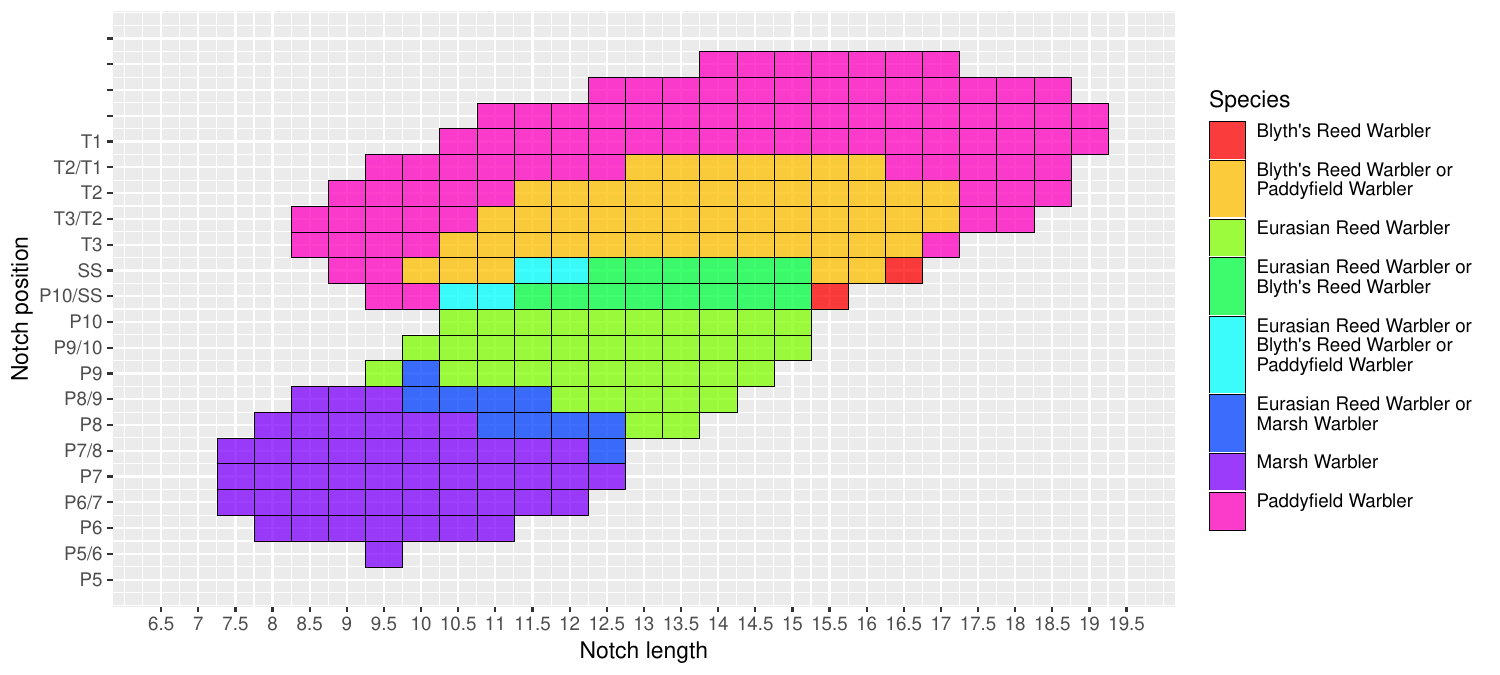} & 
    \includegraphics[width=0.5\linewidth]{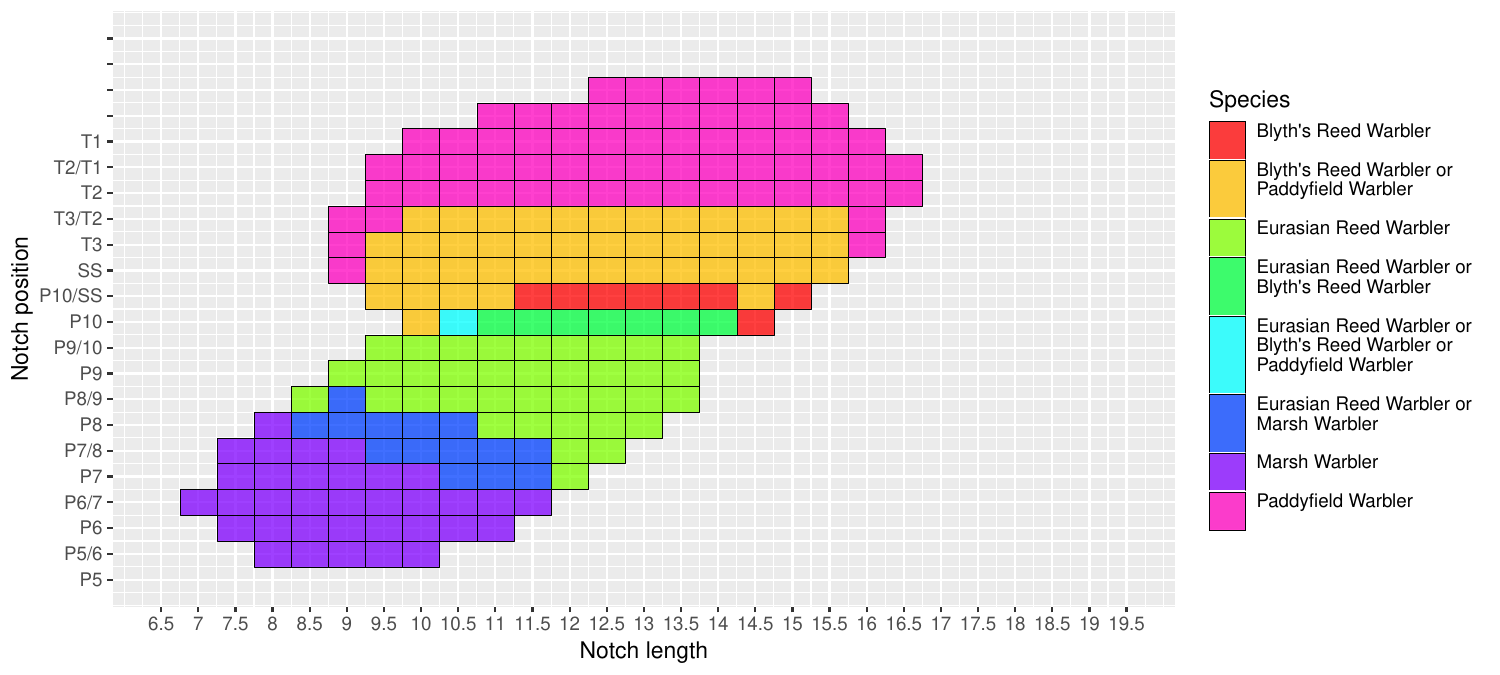} \\ [\abovecaptionskip]
     \small (a) Adult birds. & 
     \small (b) Juvenile birds.
  \end{tabular}
  \caption{Decision regions when only observing notch length and notch position.}
  \label{notchnf}
\end{figure}

\begin{figure}[h!]
  \centering
  \begin{tabular}{cc}
    \includegraphics[width=0.5\linewidth]{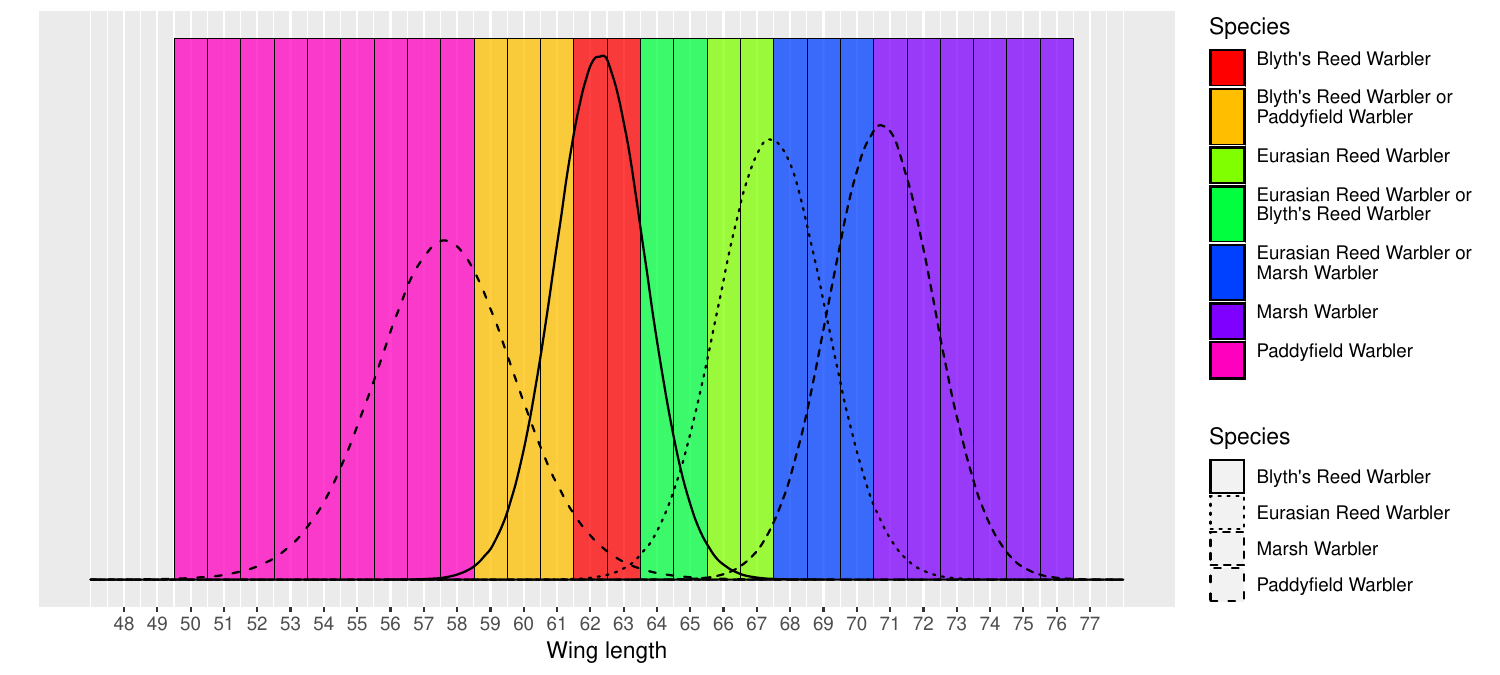} & 
    \includegraphics[width=0.5\linewidth]{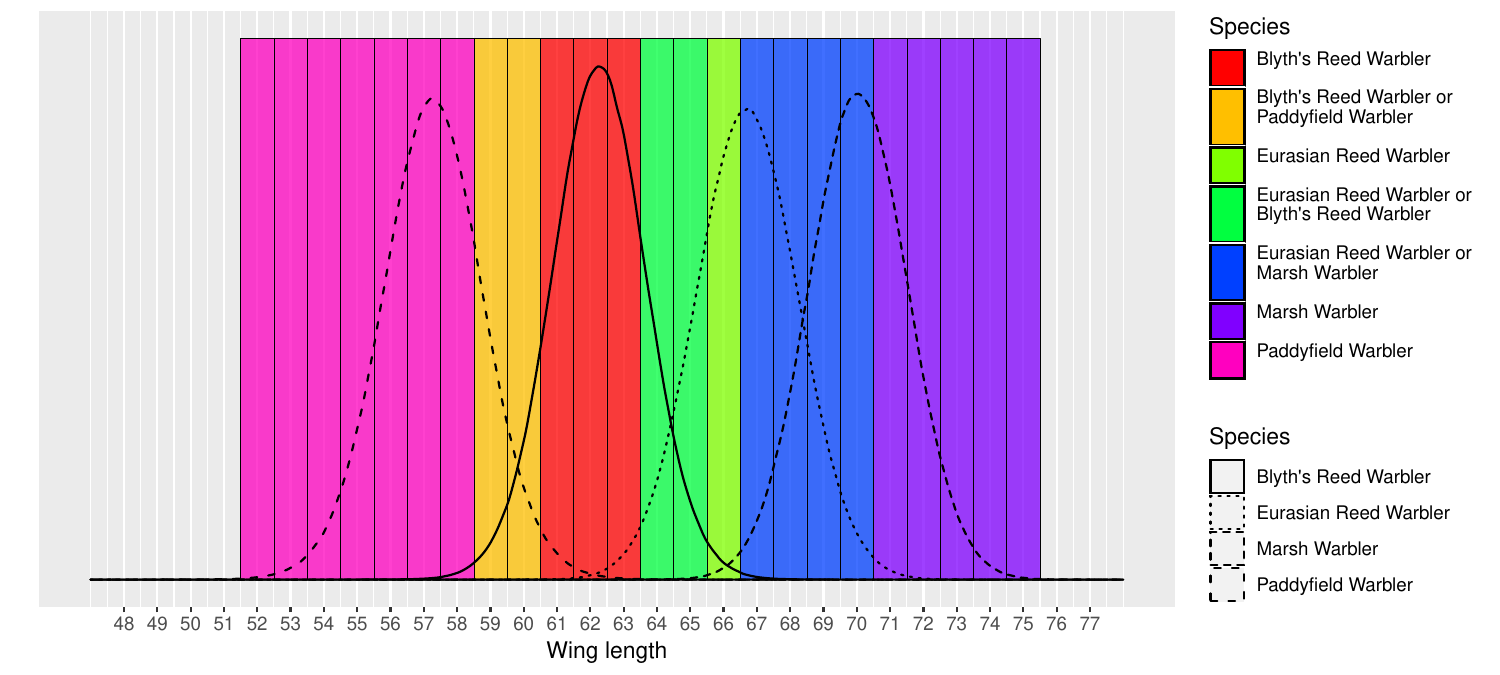} \\ [\abovecaptionskip]
     \small (a) Adult birds. & 
     \small (b) Juvenile birds.
  \end{tabular}
  \vspace{\floatsep}
  \begin{tabular}{cc}
    \includegraphics[width=0.5\linewidth]{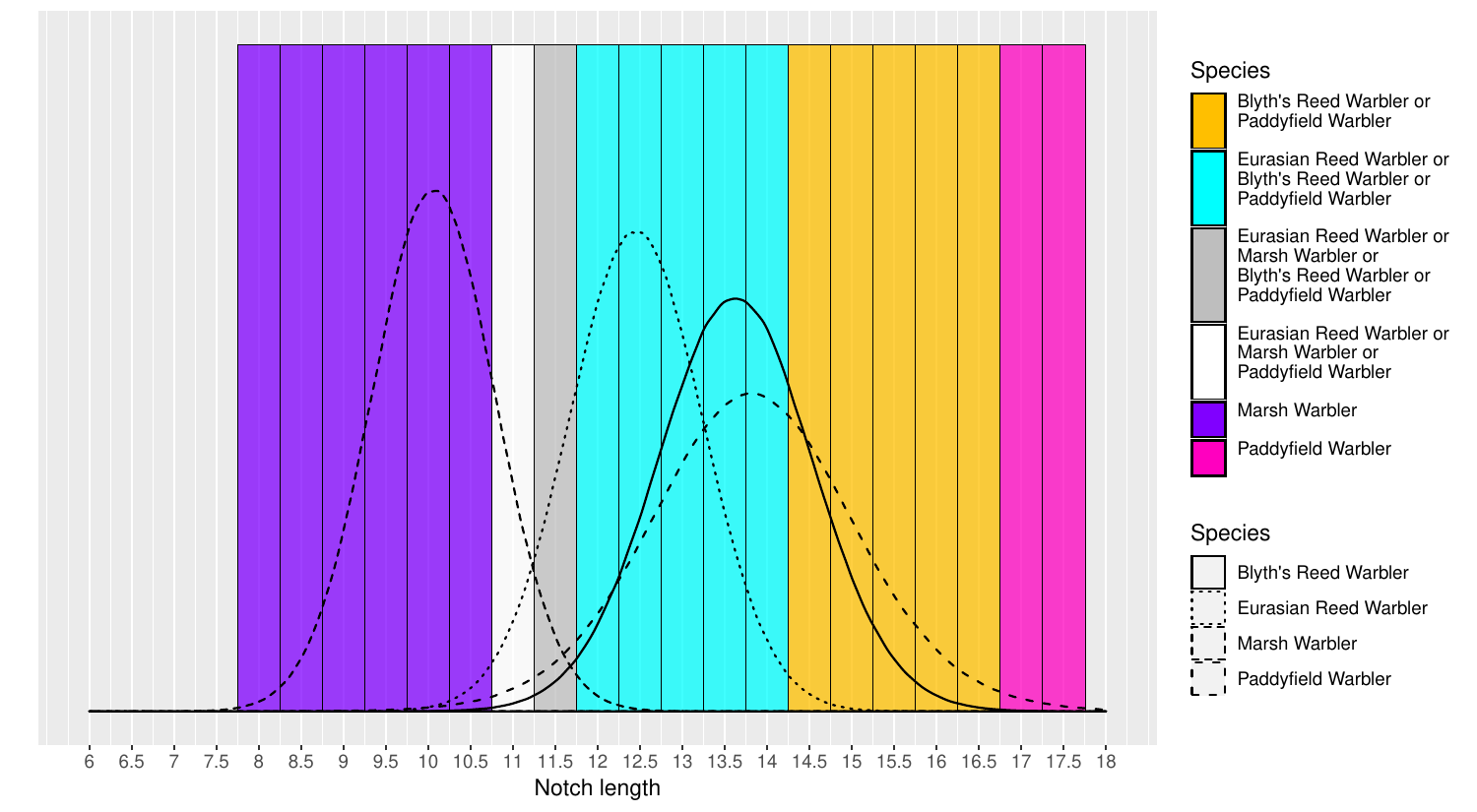} & 
    \includegraphics[width=0.5\linewidth]{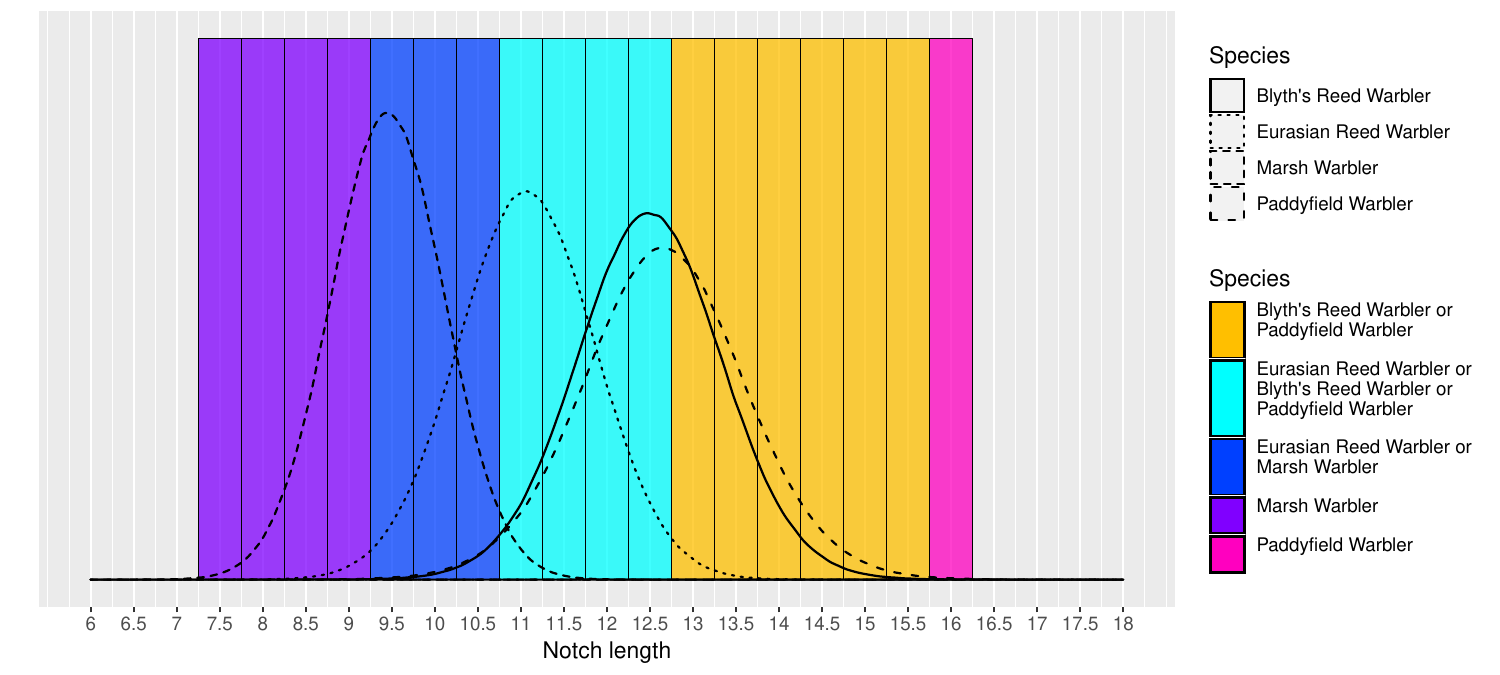} \\ [\abovecaptionskip]
     \small (c) Adult birds. & 
     \small (d) Juvenile birds.
  \end{tabular}
  \vspace{\floatsep}
  \begin{tabular}{cc}
    \includegraphics[width=0.5\linewidth]{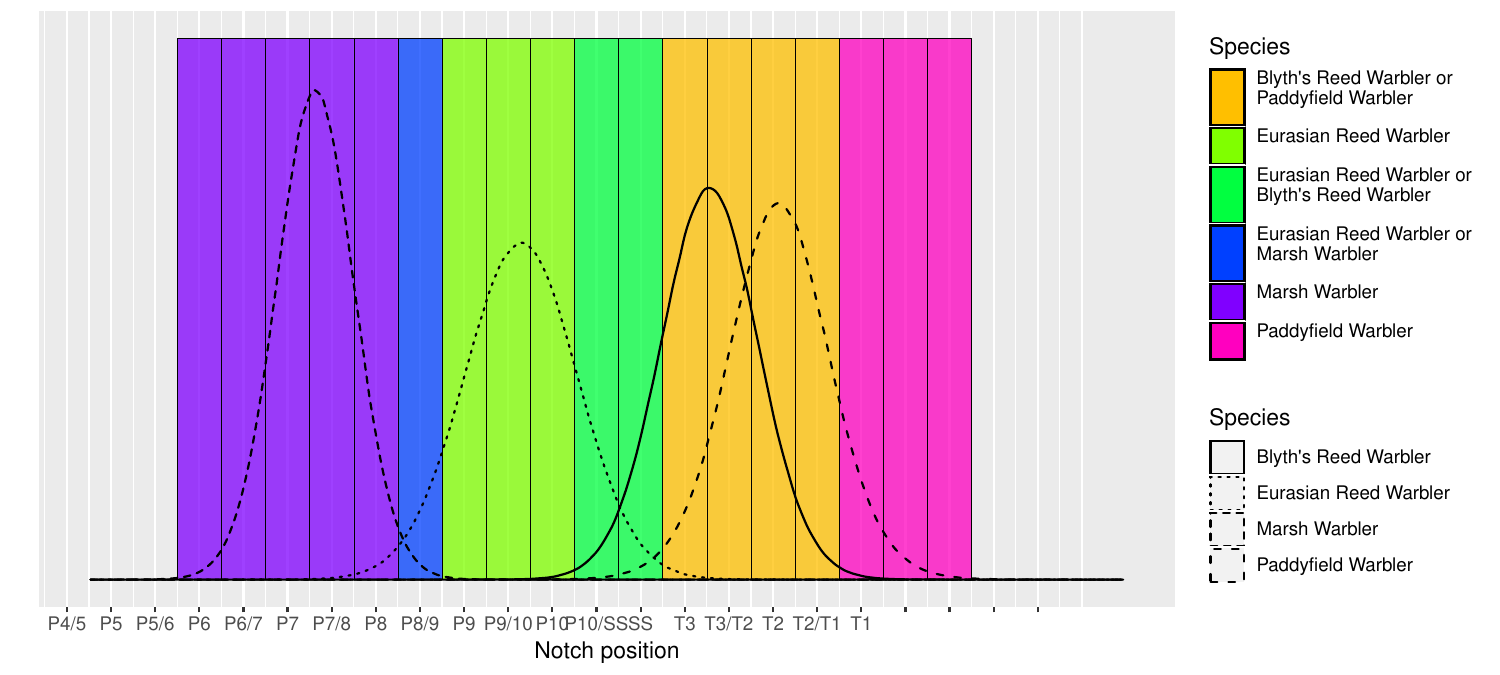} & 
    \includegraphics[width=0.5\linewidth]{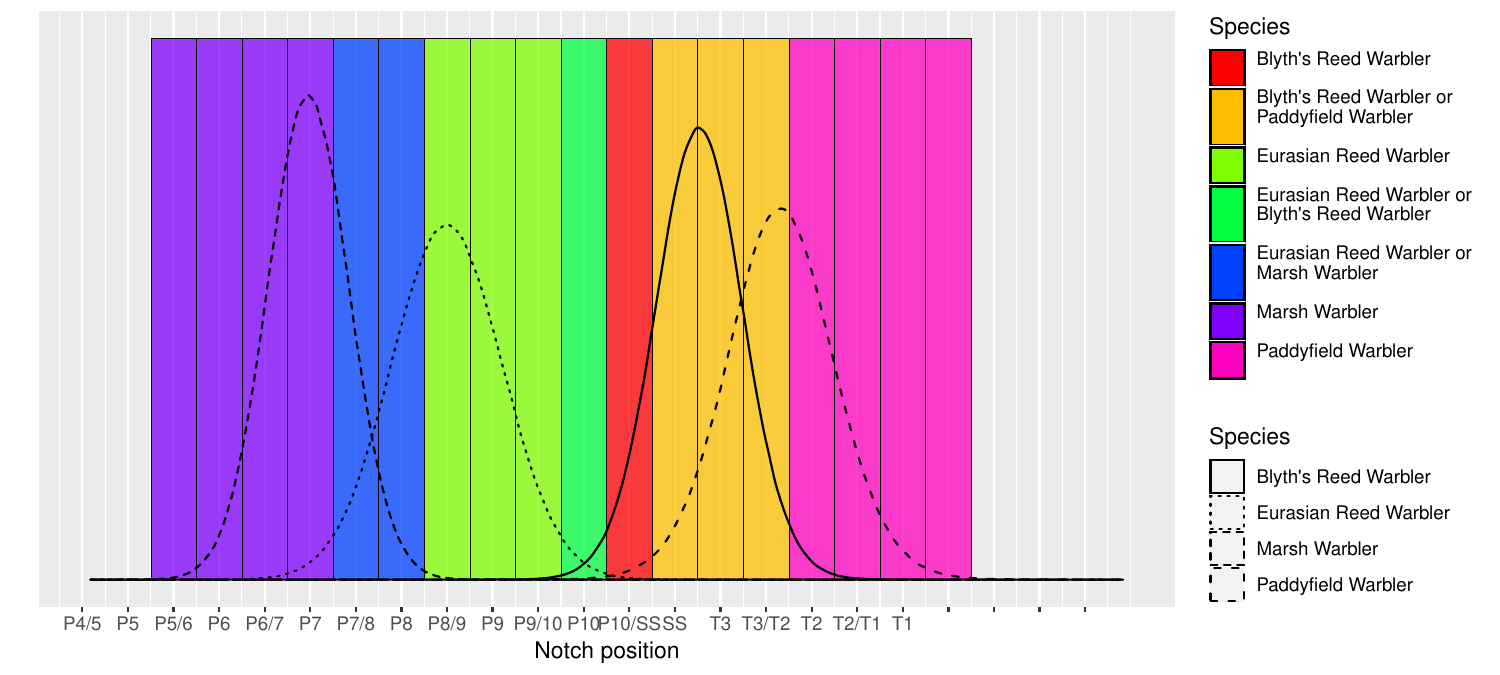} \\ [\abovecaptionskip]
     \small (e) Adult birds. & 
     \small (f) Juvenile birds.
  \end{tabular}
  \caption{In (a) and (b), decision regions are shown when only \textnormal{wing length} is observed; 
  in (c) and (d) decision regions are shown when only \textnormal{notch length} is 
  observed; and in (e) and (f) decision regions are shown when only 
  \textnormal{notch position} is observed. In all plots, kernel 
  density estimates of each aposteriori trait distribution for each 
  species is shown with black lines of different types. The plot highlights 
  the larger degree of separation in the traits \textnormal{wing length} 
  and \textnormal{notch position}.}
  \label{1d}
\end{figure}

\clearpage
\newpage

% AOS,AOAS: If there are supplements please fill:
%\begin{supplement}[id=suppA]
%  \sname{Supplement A}
%  \stitle{Title}
%  \slink[doi]{10.1214/00-AOASXXXXSUPP}
%  \sdatatype{.pdf}" 
%  \sdescription{Some text}
%\end{supplement}

\newpage
\bibliography{refs}

\end{document}